\documentclass[conference]{IEEEtran}
\IEEEoverridecommandlockouts
\usepackage[dvipsnames]{xcolor}
\usepackage{cite}
\usepackage{amsmath,amssymb,amsfonts}
\usepackage{algorithmic}
\usepackage{graphicx}
\usepackage{textcomp}
\usepackage{amsmath,amsfonts,amssymb,mathtools} 
\usepackage{amsthm}
\usepackage{bm}
\usepackage{tikz}
\usetikzlibrary{arrows,decorations.pathmorphing,automata,positioning,hobby,patterns,patterns,shapes.multipart}
\usepackage{tabularray}
\UseTblrLibrary{booktabs}
\usepackage{arydshln}
\usepackage{subcaption}

\usepackage{wasysym}
\usepackage{amscd}
\usepackage{comment}
\usepackage{listings}
\usepackage{multirow}
\usepackage{amsmath}
\usepackage{graphicx}
\usepackage{wrapfig}
\usepackage{orcidlink}
\usepackage{ltl}
\usepackage{enumitem}
\usepackage{orcidlink}

\def\BibTeX{{\rm B\kern-.05em{\sc i\kern-.025em b}\kern-.08em
    T\kern-.1667em\lower.7ex\hbox{E}\kern-.125emX}}

\makeatletter
\def\moverlay{\mathpalette\mov@rlay}
\def\mov@rlay#1#2{\leavevmode\vtop{%
		\baselineskip\z@skip \lineskiplimit-\maxdimen
		\ialign{\hfil$\m@th#1##$\hfil\cr#2\crcr}}}
\newcommand{\charfusion}[3][\mathord]{
	#1{\ifx#1\mathop\vphantom{#2}\fi
		\mathpalette\mov@rlay{#2\cr#3}
	}
	\ifx#1\mathop\expandafter\displaylimits\fi}
\makeatother

\newcommand{\cupdot}{\charfusion[\mathbin]{\cup}{\cdot}}

\newcommand{\AP}{\mathit{AP}}

\newcommand{\Cause}{\mathsf{C}}
\newcommand{\Effect}{\mathsf{E}}

\newcommand{\In}{\mathit{I}}
\newcommand{\Out}{\mathit{O}}

\newcommand{\true}[0]{\mathit{true}}
\newcommand{\false}[0]{\mathit{false}}

\newcommand{\ldot}{\mathpunct{.}}

\newcommand{\U}{\LTLuntil}
\newcommand{\X}{\LTLnext}
\newcommand{\G}{\LTLglobally}
\newcommand{\F}{\LTLeventually}
\newcommand{\R}{\LTLrelease}

\newcommand{\traces}{\mathit{traces}}

\renewcommand{\models}{\vDash}
\newcommand{\nmodels}{\nvDash}

\newcommand{\donotshow}[1]{}

\def\TT{\mathcal{T}}
\def\OO{\mathcal{O}}
\def\KK{\mathcal{K}}

\def\NN{\mathbb{N}}
\def\RR{\mathbb{R}}
\def\lim{{\varprojlim}}

\def\->{\rightarrow}

\newtheorem{theorem}{Theorem}
\newtheorem*{theorem*}{Theorem}

\newtheorem{lemma}{Lemma}

\newtheorem{corollary}{Corollary}
\newtheorem{proposition}{Proposition}

\newtheorem{definition}{Definition}
\newtheorem{example}{Example}
\newtheorem{remark}{Remark}

\newcommand{\pref}{\mathit{pref}}

\begin{document}

\title{Closure and Complexity of Temporal Causality
}

\author{\IEEEauthorblockN{Mishel Carelli, Bernd Finkbeiner, and Julian Siber}
\IEEEauthorblockA{
\textit{CISPA Helmholtz Center for Information Security}\\
Saarbrücken, Germany \\
\{mishel.carelli, finkbeiner, julian.siber\}@cispa.de}
}

\maketitle

\begin{abstract}
Temporal causality defines what property causes some observed temporal behavior (the effect) in a given computation, based on a counterfactual analysis of similar computations. In this paper, we study its closure properties and the complexity of computing causes. For the former, we establish that safety, reachability, and recurrence properties are all closed under causal inference: If the effect is from one of these property classes, then the cause for this effect is from the same class. We also show that persistence and obligation properties are not closed in this way. These results rest on a topological characterization of causes which makes them applicable to a wide range of similarity relations between computations. Finally, our complexity analysis establishes improved upper bounds for computing causes for safety, reachability, and recurrence properties. We also present the first lower bounds for all of the classes.
\end{abstract}

\begin{IEEEkeywords}
Automata, counterfactual reasoning, infinite words, temporal properties, topology
\end{IEEEkeywords}

\section{Introduction}
Temporal causality is a flavor of counterfactual reasoning that causally relates temporal properties of a given system computation. Given, for instance, the system computation $$\{\mathit{start}\}\{\mathit{request_1}\}\{\mathit{request_2}\}\{\mathit{failure}\}^\omega \enspace,$$
temporal causality can tell us whether the property $\LTLnext \mathit{request_1}$ or $\LTLnext \LTLnext\mathit{request_2}$ is the cause for the property $\LTLdiamond\LTLsquare \mathit{failure}$.
It generalizes the concept of actual causality~\cite{HalpernP01,Halpern16} to symbolic temporal properties and provides a logician's lens to study reasoning used in a plethora of applications such as explaining verification results~\cite{ChocklerHK08,BeerBCOT09,CoenenDFFHHMS22}, attribution of blame in multi-agent systems~\cite{Datta0KSS15} and explainable AI~\cite{MothilalMTS21,ChocklerH24}.

According to the theory~\cite{CoenenFFHMS22}, a causal relationship holds between two properties on a given computation, with respect to a given similarity relation, if both properties are satisfied by the computation, the most similar computations that do not satisfy the cause property do not satisfy the effect property either, and the cause property is the minimal set that qualifies for the previous two conditions. Finkbeiner et al.~\cite{FinkbeinerFMS24} have recently presented an intuitive order-theoretic reformulation of this: Causes are exactly the largest downward closed set of system computations that satisfy the effect, or -- more informally speaking -- they describe the set of computations most similar to the given observed computation that continuously satisfy the effect property. With this reformulation, they show that $\omega$-regular effects imply $\omega$-regular causes by giving an algorithm that synthesizes the cause property as a nondeterministic Büchi automaton from a system, computation and effect property with respect to a given (effectively $\omega$-regular) similarity relation. Besides showing that $\omega$-regular properties are in this way \emph{closed under causal inference}, this construction also gives an upper bound on the size of the causal automaton that is roughly exponential in the system size and doubly exponential in the effect size (see Table~\ref{fig:complexity_results} for the exact complexity). This blow-up mainly stems from Büchi complementation that is performed twice during the cause synthesis algorithm. In this work, we spin the theoretical aspects of this problem further and conduct a detailed investigation of property classes that are closed under temporal causality, as well as the complexity of constructing temporal causes as automata.

\subsection{Closure Under Temporal Causality}

While the closure of $\omega$-regular properties under causal inference is of high practical significance because it suggest a general algorithm for, e.g., constructing explanations of model-checking counterexamples, it also raises intriguing philosophical questions regarding the temporal structure of cause-effect relationships in reactive systems that interact with their environment over a possibly infinite duration. For instance, a fundamental temporal aspect of causal relationships between events is that the cause happens before the effect. A similar trait holds for temporal causality: If the effect is given as a temporal logic formula containing $n$ $\LTLnext$-operators (which is a way of describing a set of concrete events), then the cause can be described by a formula containing at most $n$ $\LTLnext$-operators~\cite{BeutnerFFS23}. Hence, the events described by the cause are guaranteed to happen earlier or at the same time as the events described by the effect. Besides demonstrating this temporal aspect of causal relationships between events, this also gives a complete cause-synthesis algorithm for this fragment based on enumeration~\cite{BeutnerFFS23}.

In this paper, we go beyond events as described by the fragment containing only $\LTLnext$-operators to general temporal properties. In this general setting, a comparable notion of temporal precedence does not exist, as general temporal properties may place conditions over a full infinite word: For instance, it is impossible to say that $\LTLsquare \LTLdiamond a$ happens before $\LTLsquare \LTLdiamond b$. Therefore, we study closure properties along the lines of Manna and Pneuli's hierarchy of temporal properties~\cite{MannaP89}. The hierarchy organizes temporal properties into classes based on language-theoretic considerations that have a tight connection to concepts in topology, temporal logic, and automata theory. We analyze for which classes membership of the effect property implies membership in the same class for the resulting cause property. Our results and previous ones are illustrated along a recapitulation of the hierarchy of $\omega$-regular properties in Figure~\ref{fig:hierarchy_results}. Note that we follow Manna and Pneuli in using LTL operators for the sake of illustration, but mean all $\omega$-regular properties in the respective classes. On the highest level the class of reactivity properties, which is equivalent to the class of $\omega$-regular properties~\cite{MannaP89}, is closed under causality as already shown by Finkbeiner et al.~\cite{FinkbeinerFMS24}. One level below, things are less clear-cut: While every recurrence effect indeed has a recurrence cause, we show that this is not the case for the class of persistence properties. On the lower levels, we have that safety and guarantee properties are both closed under causality, while obligation properties, which correspond to Boolean %CNF-like 
combinations of properties from these classes as indicated in the figure, are not closed in the same way. Interestingly, the unnamed class of properties that are both safety and guarantee properties corresponds exactly to the fragment containing only $\LTLnext$-operators for which Beutner et al.\ have shown closure under causality~\cite{BeutnerFFS23}. Notably, our results require only general assumptions on the similarity relation, such that they can accommodate different relations that may be desirable in different problem settings.

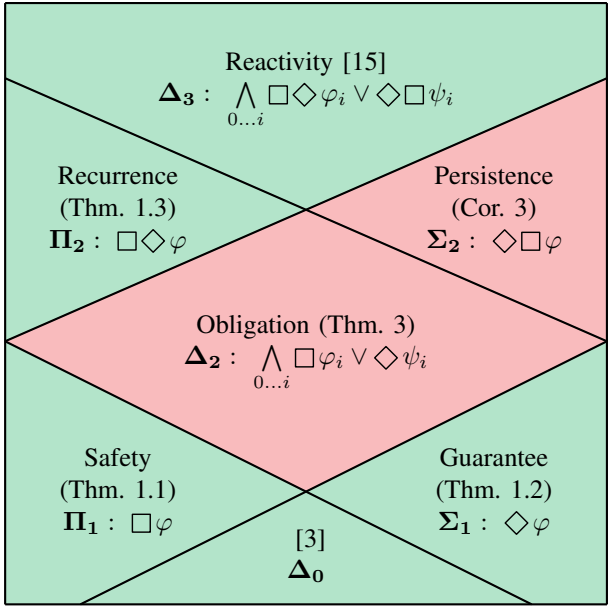
\begin{figure}
    \centering
    \normalsize
    \begin{tikzpicture}
        \draw [fill=Green, opacity=0.3,text centered]
       (0,0) -- (0,3.5) -- (4,1.5) -- (1,0) -- cycle;
        \draw [fill=Red, opacity=0.3]
       (4,5.25) -- (8,7) -- (8,3.5) -- cycle;
        \draw [fill=Green, opacity=0.3]
       (4,1.5) -- (8,3.5) -- (8,0) -- (7,0) -- cycle;
        \draw [fill=Green, opacity=0.3]
       (1,0) -- (4,1.5) -- (7,0) -- cycle;
        \draw [fill=Green, opacity=0.3]
       (0,3.5) -- (4,5.25) -- (0,7) -- cycle;
        \draw [fill=Green, opacity=0.3]
       (0,7) -- (0,8) -- (8,8) -- (8,7) -- (4,5.25) -- cycle;
        \draw [fill=Red, opacity=0.3]
       (0,3.5) -- (4,5.25) -- (8,3.5) -- (4,1.5) -- cycle;
        \draw[black,thick] (0,0) -- (8,0);
        \draw[black,thick] (8,0) -- (8,8);
        \draw[black,thick] (8,8) -- (0,8);
        \draw[black,thick] (0,8) -- (0,0);
        \draw[black,thick] (1,0) -- (8,3.5);
        \draw[black,thick] (0,3.5) -- (7,0);
        \draw[black,thick] (0,3.5) -- (8,7);
        \draw[black,thick] (0,7) -- (8,3.5);
        \node[align=center] at (6.5,5.25) {Persistence\\ (Cor.~\ref{nclosure_persistence})\\ $\boldsymbol{\Sigma_2}: \ \LTLdiamond \LTLsquare \varphi$ };
        \node[align=center] at (1.5,5.25) {Recurrence\\ (Thm.~\ref{hierarchy_preservation}.\ref{thm13})\\ $\boldsymbol{\Pi_2}: \  \LTLsquare\LTLdiamond \varphi$};
        \node[align=center] at (6.5,1.5) {Guarantee\\ (Thm.~\ref{hierarchy_preservation}.\ref{thm12})\\$\boldsymbol{\Sigma_1} : \ \LTLdiamond \varphi$ } ;
        \node[align=center] at (1.5,1.5) {Safety\\ (Thm.~\ref{hierarchy_preservation}.\ref{thm11})\\$\boldsymbol{\Pi_1} :\ \LTLsquare \varphi$};
        \node[align=center] at (4,3.35) {Obligation (Thm.~\ref{nclosure_obligation})\\$\boldsymbol{\Delta_2}: \ \bigwedge\limits_{0 \ldots i} \LTLsquare \varphi_i \lor \LTLdiamond \psi_i$};
        \node[align=center] at (4,6.85) {Reactivity~\cite{FinkbeinerFMS24}\\$\boldsymbol{\Delta_3} : \ \bigwedge\limits_{0 \ldots i} \LTLsquare\LTLdiamond \varphi_i \lor \LTLdiamond\LTLsquare \psi_i$};
        \node[align=center] at (4,0.65) {\cite{BeutnerFFS23} \\ $\boldsymbol{\Delta_0}$};
    \end{tikzpicture}
    \caption{Results on closure under causality of property classes in Manna and Pneuli's hierarchy of temporal properties~\cite{MannaP89}. Classes colored green are closed under causal inference, while classes colored in red are not. Note that $\varphi$ and $\psi$ are formulas containing no future operators.}
    \label{fig:hierarchy_results}
\end{figure}

\subsection{Complexity of Temporal Causality}

Our study of closure properties along the lines of the hierarchy of temporal properties suggest possible improvements for the synthesis of temporal causes from effects belonging to certain fragments. For instance, since we now know that a cause for a safety effect is itself a safety property, opting for a representation via bad prefixes promises a cheaper complementation operation than in the general case. Therefore, we conduct a detailed inquiry into the complexity of constructing temporal causes for effects from the varying classes. 

The main results of this inquiry are shown in Figure~\ref{fig:complexity_results}. The exact upper bound obtained from the cause synthesis algorithm of Finkbeiner et al.~\cite{FinkbeinerFMS24} is also listed in the table. As one of our main results, we can show that there is a family of problems where the cause automaton scales doubly exponential in the size of the effect. While there is still a logarithmic gap between our lower bounds and the upper bound of Finkbeiner et al.~\cite{FinkbeinerFMS24}, this demonstrates that the number of exponents of the upper bounds is already optimal. We have a tight bound for the size of the cause with respect to the size of the system, which is exponential with an additional logarithmic factor in the exponent. These results are particularly valuable because in practical instances such as explaining model-checking results, the system size tends to be much bigger than the effect size. 

We can also confirm our initial intuition regarding classes of effects for which cause automata can be synthesized more efficiently than with the general algorithm. As may be expected, this concerns the two classes on the lower level of the hierarchy: guarantee and safety properties. Since properties from these classes can effectively be described by finite word automata for good and bad prefixes, respectively, it is also possible to use the cheaper complementation operations for finite word automata in the general algorithm. We show that this approach results in an upper bound on the size of the cause automaton absent of logarithmic factors, and is mirrored by tight lower bounds in the size of the system and the effect. While the improvement in the upper bound is only by a logarithmic factor, this is still of high practical significance because the classes of guarantee and safety properties are ubiquitous in verification problems.

\subsection{Outline}

The paper is structured as follows. We first establish some basic preliminaries (Section~\ref{sec:preliminaries}). We then introduce background on temporal causality in Section~\ref{sec:temporal_causality}. Our contributions are then split into two main parts: In Section~\ref{closure_section} we establish our results on closure under causal inference using a topological argument, and in Section~\ref{sec:complexity} we prove lower and upper bounds on the size of causes. We discuss related results in Section~\ref{sec:rel_work} and end with a short summary and outlook on possible future applications of our results (Section~\ref{sec:summary}).

\begin{table*}
  \caption{Highlights of our analysis of upper and lower bounds on the size of cause automata synthesized from a system $\mathcal{T}$, computation $\pi$, a fixed similarity relation $\leq$ and an effect $E$, which belongs to a certain class as outlined in the first column.}\label{fig:complexity_results}
  \centering
        \bgroup
        \def\arraystretch{1.25}
        \setlength\tabcolsep{2em}
        \normalsize
		\begin{tabular}{cccc}
			\toprule
			\textbf{Effect Class} & \textbf{Lower bound in $\boldsymbol{E}$} & \textbf{Lower bound in $\boldsymbol{\mathcal{T}}$} & \textbf{Upper bound} \\
			\midrule
			Reactivity & $2^{2^{\Omega(|E|)}}$ (Thm.~\ref{lower_effect}) & $2^{\Omega(|\mathcal{T}| \cdot \log(\mathcal{T}|))}$ (Thm.~\ref{lower_system}.\ref{thm51})  & $|\pi| \cdot 2^{2^{\mathcal{O}(|E| \cdot \log(|E|))} \cdot |\mathcal{T}| \cdot \log(|\mathcal{T}|)}$~\cite{FinkbeinerFMS24} \\
			\midrule
			Persistence & $2^{2^{\Omega(|E|)}}$ (Thm.~\ref{lower_effect}) & $2^{\Omega(|\mathcal{T}| \cdot \log(\mathcal{T}|))}$ (Thm.~\ref{lower_system}.\ref{thm51})  & $|\pi| \cdot 2^{2^{\mathcal{O}(|E| \cdot \log(|E|))}\cdot |\mathcal{T}| \cdot \log(|\mathcal{T}|)}$~\cite{FinkbeinerFMS24} \\
			\midrule
			Recurrence & $2^{2^{\Omega(|E|)}}$ (Thm.~\ref{lower_effect}) & $\Omega(3^{|\mathcal{T}|})$ (Thm.~\ref{lower_system}.\ref{thm52}) & $|\pi| \cdot 3^{2^{\mathcal{O}(|E| \cdot \log(|E|))}\cdot |\mathcal{T}|}$ (Thm.~\ref{upper_NCW})\\
			\midrule
			Safety & $2^{2^{\Omega(|E|)}}$ (Thm.~\ref{lower_effect}) & $\Omega(2^{|\mathcal{T}|})$ (Thm.~\ref{lower_system}.\ref{thm53}) & $|\pi| \cdot 2^{2^{\mathcal{O}(|E|)}\cdot |\mathcal{T}|}$ (Thm.~\ref{upper_safety})\\
			\midrule
			Guarantee & $2^{2^{\Omega(|E|)}}$ (Thm.~\ref{lower_effect}) & $\Omega(2^{|\mathcal{T}|})$ (Thm.~\ref{lower_system}.\ref{thm53}) & $|\pi| \cdot 2^{2^{\mathcal{O}(|E|)}\cdot |\mathcal{T}|}$ (Thm.~\ref{upper_guarantee})\\
			\bottomrule
		\end{tabular}
        \egroup
\end{table*}

\section{Preliminaries}\label{sec:preliminaries}

We recall some general background on infinite words, temporal properties, automata, and temporal logic.

\subsection{Words and Properties}
We model computations by words over some \emph{alphabet} $\Sigma$. If not discussed explicitly, we assume $\Sigma$ to be finite. A \emph{word} $\pi = \pi_0 \pi_1 \ldots $ then is a sequence of letters $\pi_i \in \Sigma$ from the alphabet. For some word $\pi$, its prefix of length $n$ is denoted by $\pi_{(n)}$. The set of all prefixes of a given length is defined as $\pref(\pi) := \{ \pi_{(n)} \mid n\in \NN \}.$ We denote by $\Sigma^*$ the set of finite words and by $\Sigma^\omega$ the set of infinite words. $\Sigma^\infty = \Sigma^* \cup \Sigma^\omega$ is the set of both finite and infinite words. A \emph{language} $L$ is a subset of $\Sigma^\infty$. We denote by $\overline{L}$ the complement of a language $L$. We model \emph{properties} as finite or infinite word languages $L \subset \Sigma^*$ or  $L \subset \Sigma^\omega$, respectively. For $L\subseteq \Sigma^\omega$ the set of prefixes of words from $L$ of length $n$ is denoted as $\pref_n(L) := \{ \pi_{(n)}  \mid \pi \in L\}$ and the set of all prefixes of words in $L$ is denoted as:
$\pref(L) = \{ \pi_{(n)} \mid \pi\in L, n\in \NN\}$.

We now recall Manna and Pneuli's formal categorization of infinite word languages based on construction rules from finite world languages, which is illustrated in Figure~\ref{fig:hierarchy_results}.

\begin{definition}[Manna and Pneuli~\cite{MannaP89}]
    An infinite word language is a safety, guarantee, recurrence or persistence property, if the following holds:
    \begin{itemize}
        \item An infinite word language $L$ is a \textbf{safety property} if there exists a finite words language $\Phi$ such that $L$ consists of all infinite words $\pi$ such that every prefix of $\pi$ is in $\Phi$.

        \item An infinite word language $L$ is a \textbf{guarantee property} if there exists a finite words language $\Phi$ such that $L$ consists of all infinite words $\pi$ such that there exists a prefix of $\pi$ that is in $\Phi$.

        \item An infinite word language $L$ is a \textbf{recurrence property} if there exists a finite words language $\Phi$ such that $L$ consists of all infinite words $\pi$ such that infinitely many prefixes of $\pi$ are in $\Phi$.

        \item An infinite word language $L$ is a \textbf{persistence property} if there exists a finite words language $\Phi$ such that $L$ consists of all infinite words $\pi$ such that finitely many prefixes of $\pi$ are in $\Phi$.

        \item An infinite word language $L$ is an \textbf{obligation property} if it is a Boolean combination of safety and guarantee properties.

        \item An infinite word language $L$ is a \textbf{reactivity property} if it is a Boolean combination of recurrence and persistence properties.
    \end{itemize}
\end{definition}

Moreover, the obligation class is precisely the intersection of the recurrence and persistence classes, and the reactivity class contains all $\omega$-regular properties \cite{MannaP89}.

\subsection{Automata and Systems}
An automaton is a tuple $A = (Q,\Sigma, Q_0, F, \Delta)$, where $Q$ denotes a finite set of \emph{states}, $\Sigma$ is an alphabet, $Q_0 \subseteq Q$ is a set of \emph{initial states}, $F\subseteq Q$ is the set of \emph{accepting states}, and $\Delta: (Q \times \Sigma) \times Q$ is a \emph{transition relation} that maps a state and a letter to a set of possible successor states. For an automata $A$ the number of states of $A$ is denoted as $|A|$. An automaton $A = (Q,\Sigma, Q_0, F, \Delta)$ is \emph{deterministic} if the transition relation $\Delta$ is a function and $Q_0$ is a singleton, otherwise it is a \emph{nondeterministic} automaton. A \emph{run} of $A$ on a \emph{word} $\pi = \pi_0\pi_1 \ldots \in \Sigma^{\infty}$ is a sequence $ r = q_0q_1\ldots$ of states $q_i \in Q$ with $q_0 \in Q_0$ and $\big((q_i,w_i),q_{i+1}\big) \in \Delta$ for all $i$. \emph{Universal automata} are nondeterminstic automata that accept a word $\pi$ if all of the runs on $\pi$ fulfill the automatons acceptance condition, for all other automata only one run needs to fulfill this condition. In a \emph{finite word automaton} the acceptance condition for a run $ r = q_0q_1\ldots q_k$ is that $q_k \in F$. We use the shorthands DFW, NFW and UFW for determinstic, nondeterminstic and universal finite word automata, respectively, and analogous shorthands for all other automata types. In a \emph{Büchi word automaton} (DBW, etc.), an infinite run $ r = q_0q_1\ldots$ fulfills the acceptance condition if there exist infinitely many $i \in \mathbb{N}$ such that $q_i \in F$. Lastly, in a \emph{Co-Büchi word automaton} (DCW, etc.), an infinite run $ r = q_0q_1\ldots$ is accepting if all states appearing infinitely often are not in $F$, i.e., there is an $j \in \mathbb{N}$ such that all $i > j $ have $q_i \notin F$. The \emph{language} $L(A)$ of an automaton $A$ is the set of all words that have an accepting run. An infinite word language $L \subset \Sigma^\omega$ is $\omega$-regular, if it is recognized by a nondeterministic Büchi word automaton (NBW) $A$ such that we have $L(A) = L$. 

A \emph{systems} is a tuple $\mathcal{T} = (S, s_0, \AP, \delta, l)$ where 
$S$ is a finite set of \emph{states}, $s_0 \in S$ is the \emph{initial state}, $\AP = \In \cupdot \Out$ consists of \emph{inputs} $\In$ and \emph{outputs} $\Out$, $\delta: S \times 2^I \rightarrow 2^S$ is the \emph{transition function} describing the successor states for some state and input, and $l: S \rightarrow 2^\Out$ is the \emph{labeling function} labeling each state with a set of outputs. 
A \emph{trace} of $\mathcal{T}$ is an infinite sequence $\pi = \pi_0 \pi_1 \ldots \in (2^\AP)^\omega$, with $\pi_i = A \cup l(s_{i+1})$ for some $A \subseteq I$ and $s_{i+1} \in \delta(s_i,A)$ for all $i \geq 0$. Note that the label of the initial state is omitted in the first position. 
$\mathit{traces}(\mathcal{T})$ is the set of all traces of $\mathcal{T}$. A \emph{zipped trace} of the three traces $\pi^{0,1,2}$ is then defined as $\mathit{zip}(\pi^0,\pi^1,\pi^2)_i = \{a^k \, | \, a \in \pi_i^k\}$, i.e., we construct the zipped trace from disjoint unions of the positions of the three traces, where inputs and outputs from the traces $\pi^{0,1,2}$ are distinguished through superscripts. We also define projection and equivalence on traces: for $A,B \subseteq I \cup O$ and traces $\pi,\pi'$, let $A|_B = A \cap B$, $\pi|_B = \pi_0|_B\pi_1|_B\ldots$ and $\pi =_A \pi'$ iff $\pi|_A = \pi'|_A$.

\subsection{Linear-time Temporal Logic} We will use \emph{Linear-time Temporal Logic} (LTL)~\cite{Pnueli77} when we want to describe a property more conveniently than with automata in this paper (even though not every $\omega$-regular property can be expressed this way). The grammar for LTL formulas is as follows, where $a \in \Sigma$:
\begin{equation*}
\varphi \Coloneqq a \mid \neg \varphi \mid \varphi \land \varphi \mid \LTLnext \varphi \mid \varphi \LTLu \varphi \mid \LTLnext^- \varphi \mid \LTLuntil^- \varphi \enspace .
\end{equation*}
All temporal operators with a minus superscript are past operators~\cite{LichtensteinPZ85}, the others are future operators. The semantics of LTL are given as follows.
\begin{equation*}
\begin{array}{lll}
\pi,i \models a       & \text{iff } & a = \pi_i \\
\pi,i  \models \neg \varphi              & \text{iff } & \pi,i  \nmodels \varphi \\
\pi,i  \models \varphi \land \psi         & \text{iff } & \pi,i  \models \varphi \text{ and } \pi,i \models \psi \\
\pi,i \models \X \varphi                & \text{iff } & \pi,i+1 \models \varphi \\
\pi,i \models \LTLnext^- \varphi	&\text{iff }	& i > 0 \land \pi,i-1 \models \varphi\\
\pi,i \models \varphi\U\psi             & \text{iff } & \exists j \geq i \text{ such that } \pi,j \models \psi \text{ and }\\ & &\forall i \leq k < j \ldot \pi,k \models \varphi\\
\pi,i \models \varphi \LTLuntil^- \psi	& \text{iff } 	& \exists k \leq i \text{ such that } \pi,k \models \psi\text{ and }\\ & & \forall i \geq j > k: \pi,j \models \varphi \enspace.
\end{array}
\end{equation*}
A word $\pi$ \emph{satisfies} a formula $\varphi$, denoted by $\pi \models \varphi$  iff $\pi,0 \models \varphi$, i.e., the formula holds at the first position.
The \emph{language} $L(\varphi)$ of a formula $\varphi$ is the set of all traces that satisfy it. We also use the derived Boolean connectives ($\lor$, $\rightarrow$, $\leftrightarrow$) and temporal operators ($\varphi \R \psi \equiv \neg(\neg \varphi \U \neg \psi)$, $\F \varphi \equiv \true \U \varphi$, $\G \varphi \equiv \false \R \varphi$, $\LTLeventually^- \varphi \equiv \true \LTLuntil^- \varphi$, $\LTLglobally^- \varphi \equiv \lnot \LTLeventually^- \lnot \varphi$).

\section{Temporal Causality}\label{sec:temporal_causality}
We now present a comprehensive primer on previous work regarding temporal causality that is relevant to this paper. Temporal causality is concerned with analyzing the cause for some effect emerging on an observed computation of a reactive system. In particular, this observed computation can be infinite, such as obtained as a counterexample from model checking.
Informally speaking, the conditions for a causal relationship between two properties are as follows~\cite{CoenenFFHMS22}:

\begin{itemize}
    \item Both the cause property and the effect property are satisfied by the observed computation (\textbf{SAT} condition).
    \item The closest, i.e., most similar traces that do not satisfy the cause also do not satisfy the effect (\textbf{CF} condition).
    \item No subset of the cause property satisfies the previous two conditions (\textbf{MIN} condition).
\end{itemize}

More formally, the definition of such a causal relationship requires fixing a notion of closeness between system computations based on a similarity relation $\leq_\pi \, \subseteq \Sigma^\omega \times \Sigma^\omega$, which orders two traces $(\pi^1,\pi^2) \in \; \leq_\pi$ if $\pi^1$ is at least as similar to $\pi$ than $\pi^2$. Such a relation can be expressed by a (relational) temporal formula such that, for instance, $\pi^{\mathit{1}} \leq^{\mathit{subset}}_{\pi^0} \pi^{2}$ iff:
$$\mathit{zip}(\pi^0,\pi^{1},\pi^{2}) \models \G \bigwedge_{i\in I} \big((i^0 \not\leftrightarrow i^1) \rightarrow (i^0 \not\leftrightarrow i^2)\big) \enspace .$$

Note that while we allow similarity relations over the full alphabet, we are usually interested in similarity of the input sequences~\cite{CoenenFFHMS22}, such as defined by $\leq^{\mathit{subset}}$. With these similarity relations at hand, we can now recall the formal definition of temporal causality.

\begin{definition}[Temporal Cause~\cite{FinkbeinerFMS24}]\label{def:causality}
    Let $\mathcal{T}$ be a system, $\pi\in \mathit{traces}(\mathcal{T})$ a computation of the system, $\leq_\pi$ a similarity relation, and $E \subseteq (2^\AP)^\omega$ an effect property. We say that $C  \subseteq (2^{I})^\omega$ is a cause of $E$ on $\pi$ in $\mathcal{T}$ if the following holds.
	\begin{description}[leftmargin=!]
		\item[SAT:] $\forall \pi' \in \mathit{traces}(\mathcal{T}) : \pi' =_I \pi \rightarrow \pi'|_I \in C \land \pi' \in E$.
		\item[CF:] $\forall \pi' \in \overline{C} : \exists \pi'' \in \mathit{traces}(\mathcal{T}) : \pi'' \leq_\pi \pi' \land \pi'' \in \overline{E}$.
		\item[MIN:] $\nexists C' \subset C :$ $C'$ satisfies \emph{\textbf{SAT}} and \emph{\textbf{CF}}.
	\end{description}	
\end{definition}

There are two details that stand out over our earlier, informal definition: First, the \textbf{SAT} condition requires that all traces that are input-equivalent to the observed trace satisfy cause and effect, which ensures that there is no causal property in the case of nondeterminism on the observed trace. Second, the \textbf{CF} condition is realized through a $\forall\exists$-quantifier alternation because there are cases where $(\mathit{traces}(\mathcal{T}),\leq_\pi)$ is not well-founded, such that no ``closest" traces exist and the limit assumption for counterfactual reasoning is not met~\cite{Lewis73a,FinkbeinerFMS24}.

\begin{figure*}[t]
         \begin{subfigure}{.5\textwidth}
         \centering
            \begin{tikzpicture}[->,shorten >=1pt,thick,auto,node distance=2.75cm, on grid,initial text=,
			every state/.style={minimum size=20pt,inner sep=2pt}]
         \node[state,initial left](1){$\emptyset$};
         \node[state, right = of 1](2){$\emptyset$};
         \node[state, right = of 2](3){$\{o\}$};
         \node[state,above = of 2](4){$\emptyset$};
         \path[->,draw](1) edge[] node[swap]{$i$} (2)
         (2) edge[] node[swap]{$i$} (3)
         (2) edge[] node[pos=0.25]{$\lnot i$} (4);
         \path[->,draw](4) edge[] node[]{$i$} (3)
         (3) edge[loop below] node[]{$\top$} (3);
         \path[->,draw](1) edge[loop below] node[]{$\lnot i$} (1);
         \path[->,draw](4) edge[] node[swap]{$\lnot i$} (1);
         \end{tikzpicture}\hspace{2em}
         \caption{An example reactive system $\mathcal{T}$.}\label{fig:example}
     \end{subfigure}%
     \begin{subfigure}{.5\textwidth} 
        \centering
            \begin{tikzpicture}[>=latex,thick,label distance=-1pt]
                \draw[rounded corners=1mm,color=Goldenrod,fill=Goldenrod,line width=0pt,fill opacity=0.5,opacity=0.5] (5.72,0.28) -- (6.28,0.28)  -- (6.28,-3.78) -- (5.72,-3.78) -- cycle;
                
                \draw[rounded corners=1mm,color=MidnightBlue,fill=Goldenrod,line width=3pt,fill opacity=0.5] (-0.25,0.25) -- (3.0,0.25)  -- (3.0,-3.75) -- (-0.25,-3.75) -- cycle;
            
                \node[draw, circle, inner sep=2pt, fill=white] at (0,0)(s0) {};
                \node[draw, circle, inner sep=2pt, fill=white, below =1 of s0](a0) {};
                \node[draw, circle, inner sep=2pt, fill=white, below =1 of a0](b0) {};
                \node[circle, inner sep=2pt, below =1 of b0](c0) {};
                \draw[-stealth,decorate,decoration={snake,segment length=3.5pt,amplitude=0.75pt,pre length=1pt,post length=3pt}] (s0) -- (a0);
                \draw[-stealth,decorate,decoration={snake,segment length=3.5pt,amplitude=0.75pt,pre length=1pt,post length=3pt}] (a0) -- (b0);
                \draw[-stealth,decorate,decoration={snake,segment length=3.5pt,amplitude=0.75pt,pre length=1pt,post length=3pt}] (b0) -- (c0);
                \node[above = 0.25 of s0](p0){$\pi_a$};
                
                \node[draw, circle, inner sep=2pt, fill=white] at (1.5,0)(s1) {};
                \node[draw, circle, inner sep=2pt, fill=white, below =1 of s1](a1) {};
                \node[draw, circle, inner sep=2pt, fill=white, below =1 of a1](b1) {};
                \node[circle, inner sep=2pt, below =1 of b1](c1) {};
                \draw[-stealth,decorate,decoration={snake,segment length=3.5pt,amplitude=0.75pt,pre length=1pt,post length=3pt}] (s1) -- (a1);
                \draw[-stealth,decorate,decoration={snake,segment length=3.5pt,amplitude=0.75pt,pre length=1pt,post length=3pt}] (a1) -- (b1);
                \draw[-stealth,decorate,decoration={snake,segment length=3.5pt,amplitude=0.75pt,pre length=1pt,post length=3pt}] (b1) -- (c1);
                \node[above = 0.25 of s1](p1){$\pi_b$};
                
                \node[draw, circle, inner sep=2pt, fill=white] at (4.5,0)(s2) {};
                \node[draw, circle, inner sep=2pt, fill=white, below =1 of s2](a2) {};
                \node[draw, circle, inner sep=2pt, fill=white, below =1 of a2](b2) {};
                \node[circle, inner sep=2pt, below =1 of b2](c2) {};
                \draw[-stealth,decorate,decoration={snake,segment length=3.5pt,amplitude=0.75pt,pre length=1pt,post length=3pt}] (s2) -- (a2);
                \draw[-stealth,decorate,decoration={snake,segment length=3.5pt,amplitude=0.75pt,pre length=1pt,post length=3pt}] (a2) -- (b2);
                \draw[-stealth,decorate,decoration={snake,segment length=3.5pt,amplitude=0.75pt,pre length=1pt,post length=3pt}] (b2) -- (c2);
                \node[above = 0.25 of s2](p2){$\pi_c$};
                
                \node[draw, circle, inner sep=2pt, fill=white] at (6,0)(s3) {};
                \node[draw, circle, inner sep=2pt, fill=white, below =1 of s3](a3) {};
                \node[draw, circle, inner sep=2pt, fill=white, below =1 of a3](b3) {};
                \node[circle, inner sep=2pt, below =1 of b3](c3) {};
                \draw[-stealth,decorate,decoration={snake,segment length=3.5pt,amplitude=0.75pt,pre length=1pt,post length=3pt}] (s3) -- (a3);
                \draw[-stealth,decorate,decoration={snake,segment length=3.5pt,amplitude=0.75pt,pre length=1pt,post length=3pt}] (a3) -- (b3);
                \draw[-stealth,decorate,decoration={snake,segment length=3.5pt,amplitude=0.75pt,pre length=1pt,post length=3pt}] (b3) -- (c3);
                \node[above = 0.25 of s3](p3){$\pi_d$};
                
                \node[draw, circle, inner sep=2pt, fill=white] at (7.5,0)(s4) {};
                \node[draw, circle, inner sep=2pt, fill=white, below =1 of s4](a4) {};
                \node[draw, circle, inner sep=2pt, fill=white, below =1 of a4](b4) {};
                \node[circle, inner sep=2pt, fill=white, below =1 of b4](c4) {};
                \draw[-stealth,decorate,decoration={snake,segment length=3.5pt,amplitude=0.75pt,pre length=1pt,post length=3pt}] (s4) -- (a4);
                \draw[-stealth,decorate,decoration={snake,segment length=3.5pt,amplitude=0.75pt,pre length=1pt,post length=3pt}] (a4) -- (b4);
                \draw[-stealth,decorate,decoration={snake,segment length=3.5pt,amplitude=0.75pt,pre length=1pt,post length=3pt}] (b4) -- (c4);
                \node[above = 0.25 of s4](p4){$\pi_e$};

                \node[right = 0.05 of p0](l01){$\leq_{\pi_a}$};
                \node[right = 0.05 of p1](l12){$\leq_{\pi_a}$};
                \node[left = 0.8 of p2](l12d){$\cdots$};
                \node[left = 0.05 of p2](l12x){$\leq_{\pi_a}$};
                \node[right = 0.05 of p2](l23){$\leq_{\pi_a}$};
                \node[right = 0.05 of p3](l34){$\leq_{\pi_a}$};
                \node[below = 1.75 of l01](d0){$\ldots$};
                \node[below = 1.75 of l12](d1){$\ldots$};
                \node[below = 1.75 of l12x](d1x){$\ldots$};
                \node[below = 1.75 of l23](d2){$\ldots$};
                \node[below = 1.75 of l34](d3){$\ldots$};
            
            \end{tikzpicture}
         \caption{Cause $C$ as the complement of the upward closure of $\overline{E}$~\cite{FinkbeinerFMS24}.}\label{fig:closure}
     \end{subfigure} 
     \caption{Figure~\ref{fig:example} illustrated the reactive system with the input $i$ and output $o$ that is used to outline temporal causality in Example~\ref{ex:intro}. The system sets the output $o$ continuously whenever the input $i$ is enabled less than three time units apart. Figure~\ref{fig:closure} pictures a chain in $(\mathit{traces}(\mathcal{T}),\leq_{\pi_a})$ to illustrate that the cause $C$ (enclosed by the blue frame) on $\pi_a$ is the largest downward closed set of traces satisfying the effect $E$ (shaded yellow), which is the complement of the upward closure of $\overline{E}$.} 
\end{figure*}
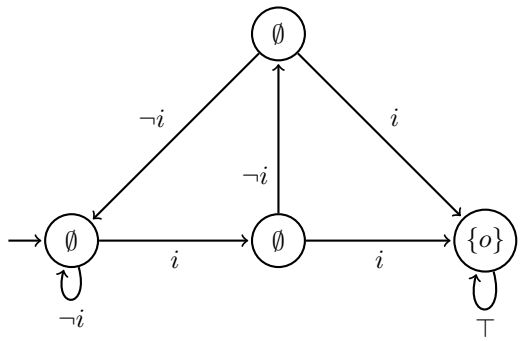
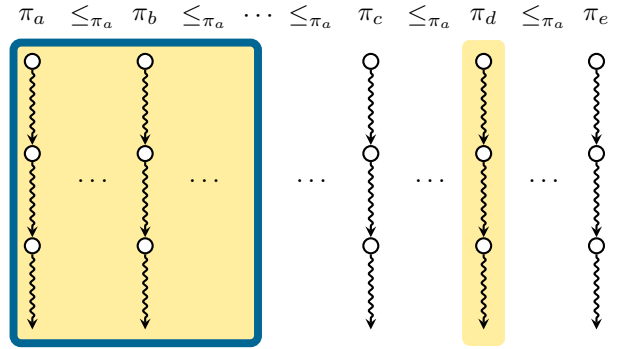

\begin{example}\label{ex:intro}
    Consider the reactive system $\mathcal{T}$ illustrated in Figure~\ref{fig:example}, the trace $\pi = \{i\}\{i,o\}\{o\}^\omega$, and the effect $E = \LTLglobally \LTLeventually o$, with similarity relation $\leq^{\mathit{subset}}$. The cause $C_\pi$ for $E$ on $\pi$ is characterized by the formula $C_\pi = L(i \land \LTLnext i)$. It is easy to see that the \textbf{SAT} condition is met, let us take a closer look at the other two. For \textbf{CF}, we can find for any $\pi' \notin C_\pi$ either $\{\}\{i\}\{\}^\omega$ or $\{i\}\{\}\{\}^\omega$ as a $\pi'' \notin E$, in particular for traces such as $\pi''' = \{i\}\{\}\{i\}\{i,o\}^\omega$ that satisfy $E$ but are less similar to $\pi$, e.g., $\{i\}\{\}\{\}^\omega \leq_\pi^{\mathit{subset}} \pi'''$. For \textbf{MIN}, we can see that restricting $C_\pi$ further in any way that satisfies \textbf{SAT} leads to violation of the \textbf{CF} condition: For instance, if we had $C_\pi' = L( i\land \LTLnext i \land \lnot \LTLnext \LTLnext i)$, then there would be the trace $\{i\}\{i,o\}\{i,o\}\{o\}^\omega \notin C_\pi'$ for which no at least as similar trace exists that does not satisfy the effect. It gets more complex when we have a trace such as $\sigma = \{i\}\{i,o\}^\omega$ where no finite number of inputs is responsible for obtaining the effect $E$. We can perform a similar analysis as above to show that $C_\sigma = L\big(\LTLdiamond (i \land \LTLnext  (i \lor \LTLnext i))\big)$ is the cause for $E$ on $\sigma$.
\end{example}

While Definition~\ref{def:causality} with its three conditions closely mirrors Halpern and Pearl's definition of actual causality~\cite{HalpernP01} by which it is inspired, Finkbeiner at al.~\cite{FinkbeinerFMS24} have shown recently that causes can be characterized much more succinctly. This is because they directly correspond to the largest downward closed set of traces in $(\mathit{traces}(\mathcal{T}),\leq_\pi)$ that satisfy the effect $E$. Vice versa, they are also the complement of the upward closure of $\overline{E}$. This stems from the balance of the \textbf{CF} and \textbf{MIN} conditions: If there was a trace in the cause that does not satisfy the effect, its upward closure could be removed -- including other traces that do in fact satisfy the effect. As pictured in Figure~\ref{fig:closure}, this results in a cause (framed by the blue border) that, in essence, describes the local continuous neighborhood of traces that are similar to the observed trace $\pi_a$ and satisfy the effect $E$ colored by the yellow area. More formally, this is captured by the following lemma.

\begin{lemma}[Finkbeiner et al.~\cite{FinkbeinerFMS24}]\label{lem:cause}
    Let $\mathcal{T}$ be a system, $\pi\in \mathit{traces}(\mathcal{T})$
    a computation of the system, $\leq_\pi$ a similarity relation, and $E \subseteq \Sigma^\omega$ an 
    effect property. If there is a cause $C$ of $E$ on $\pi$ in $\mathcal{T}$ then it is the largest downward closed set of system computations, i.e., $$C = \{\,\rho \in (2^I)^\omega \,  \mid \, \forall \sigma \in \mathit{traces}(\mathcal{T}) \ \sigma \leq_\pi \rho \rightarrow \sigma \in E \,\} \enspace.$$
    Moreover, the above set is empty iff there is no cause
\end{lemma}

Notably, Lemma~\ref{lem:cause} also means a cause is unique if it exists. In the following, we use this convenient characterization for a detailed study of the closure of property classes under causal inference and of complexity bounds on the size of automata representations for causes.

\section{Closure Under Causal Inference} \label{closure_section}

In this section, we investigate the closure under causal inference of classes in the temporal hierarchy. As can be seen in Lemma~\ref{lem:cause}, causes can be directly defined via universal quantification over the system traces to express that they are downward closed in $(\mathit{traces}(\mathcal{T}),\leq_\pi)$ and all satisfy the effect $E$. Speaking more abstractly, we are hence interested in whether properties $X$ for which membership of some word is decided based on whether \emph{all} associated words satisfy some other property $Y$ inherit the hierarchy property classes from $Y$. We introduce a more convenient yet more general concept to facilitate this abstract analysis: the \emph{universal preimage}. This concept essentially formalizes set membership based on universal quantification and makes our results applicable beyond closure under causal inference, i.e., to quantifier elimination in automata-based model checking of hyperproperties~\cite{ClarksonFKMRS14,FinkbeinerRS15} and synthesis from partial information~\cite{Kupferman2000}.

\begin{definition}
    Let $X$ and $Y$ be sets, and let $f: X \to \mathcal{P}(Y)$ be a function. The \textbf{universal preimage} of $S$ under $f$, denoted $f^{-1}_U(S)$, is the set of elements whose images under $f$ are subsets of $S$: $f^{-1}_U(S) := \{x\in X \mid f(x) \subseteq S \}$.
\end{definition}
Informally, for an element $x\in X$ the function $f$ determines the set of values of $Y$ over which our universal quantifier ranges (e.g., all traces at least as similar as $\rho$ in Lemma~\ref{lem:cause}). If all of them are included in the set $S$, then $x$ is included in the universal preimage $f^{-1}_U(S)$.

Causality can be defined via a universal preimage: The cause of $E$ on $\pi$ in $\TT$ with similarity relation $\leq$ is the universal preimage of the effect under the map that sends a trace $\pi''$ to the set of at least as similar traces $\pi'\in \TT$, such that $\pi'\leq_{\pi} \pi''$. 

The choice of similarity relation heavily influences the closure of temporal hierarchy classes under causality. Different similarity relations produce different functions $f$. Instead of proving closure under causal inference for one specific similarity relation, we identify general assumptions on the relation, and hence the function $f$, that facilitate our closure results, such that they can be transferred to a variety of use cases and similarity relations.

The primary goal for the rest of this section is to formulate and prove the closure of temporal classes under the universal preimage. Section~\ref{linguistic_closure} introduces the \emph{Universal Closure Theorem} and provides a linguistic perspective on the main definitions. Section~\ref{topological_closure} then provides a more abstract topological interpretation of the theorem and proves it using results from general topology. The last subsection (Section~\ref{sec:non-closure}) presents the concept of the existential projection and discusses the parts of the temporal hierarchy which are not closed.

\subsection{Linguistic Formulation} \label{linguistic_closure}

In general, temporal classes are not be closed under the universal preimage for an arbitrary $f: \Sigma_1^\omega \to \mathcal{P}(\Sigma_2^\omega)$. To ensure closure, we impose certain restrictions on $f$. 

Informally, to show that reachability is closed under $f_U^{-1}$, we need to show that if for some $\pi''$ every trace $\pi'\in f(\pi'')$ has a good prefix, then $\pi''$ must also have a good prefix. 

The first property of $f$, which we need, is that we can build a prefix tree of $f(\pi'')$ while reading $\pi''$. It could be the case, for example, that the second level of the tree (the set of prefixes of $f(\pi'')$ of the length $2$) cannot be determined unless the whole word $\pi''$ is observed.  If such a case, the reachability class might not be closed under the universal preimage, since we do not know if every trace from $f(\pi'')$ has a good prefix unless we observe the whole $\pi''$.

\begin{definition}
    A function $f: \Sigma_1^\omega \to \mathcal{P}(\Sigma_2^\omega)$ is called \textbf{prefix-continuous} if for every $n\in \NN$ and $\pi \in \Sigma_1^\omega$ there exists $m\in \NN$ such that, for every $\pi' \in \Sigma_1^\omega$ if $\pi'_{(m)} = \pi_{(m)}$, then
    $$\pref_n(f(\pi)) = \pref_n(f(\pi')) \enspace .$$
\end{definition}

Informally, this ensures that the prefix tree of $f(\pi'')$ can be constructed incrementally.

The second property that we need is that every infinite trace in the prefix tree corresponds to some trace in $f(\pi'')$. 

\begin{definition}
     A function $f: \Sigma_1^\omega \to \mathcal{P}(\Sigma_2^\omega)$ is called \textbf{prefix-closed} if for every $\pi \in \Sigma_1^\omega$ and $\pi' \in \Sigma_2^\omega$ :
     $$\pref(\pi') \subseteq \pref(f(\pi)) \Rightarrow \pi' \in f(\pi) \enspace .$$
\end{definition}

We already have everything that we need for the case of a finite alphabet $\Sigma_2$, but if it is infinite then we set one additional restriction. The last property that we require from $f$ is the finite branching of the prefix tree. Clearly, this is trivially satisfied for a finite alphabet.

\begin{definition}
    A function $f: \Sigma_1^\omega \to \mathcal{P}(\Sigma_2^\omega)$ is called \textbf{prefix-compact} if for every $n\in \NN$ and $\pi \in \Sigma_1^\omega$ the set $\pref_n(f(\pi))$ is finite.
\end{definition}

\begin{remark}
    If the alphabet $\Sigma_2$ is finite, then any function $f:  \Sigma_1^\omega \to \mathcal{P}(\Sigma_2^\omega)$ is prefix-compact. 
\end{remark}

Now everything is ready to formulate the main Theorem of this section, which formally states that several classes from the temporal hierarchy are closed under the universal preimage operation, for functions that satisfy the previously introduced requirements.

\begin{theorem} [Universal Closure]\label{hierarchy_preservation}
    Let $\Sigma_1$ and $\Sigma_2$ be alphabets. Assume a function $f:\Sigma_1^\omega \to \mathcal{P}(\Sigma_2^\omega)$ is \textbf{prefix-continuous, prefix-closed}, and \textbf{prefix-compact}, then:
    \begin{enumerate}
    
        \item\label{thm11} If $L$ is a safety property, then $f_U^{-1}(L)$ is a safety property.

        \item\label{thm12} If $L$ is a guarantee property, then $f_U^{-1}(L)$ is a guarantee property.

        \item\label{thm13} If $L$ is a recurrence property, then $f_U^{-1}(L)$ is a recurrence property.
    \end{enumerate}
    
\end{theorem}

The proof of Theorem~\ref{hierarchy_preservation} requires establishing some auxiliary results using a topological argument, which we formulate in Section~\ref{topological_closure}. With the theorem at hand, we can show that the general result in particular covers the previously introduced similarity relation $\leq^{subset}$, since this relation satisfies all the introduced requirements.

\begin{proposition} \label{causality_is_up}
    Let $\mathcal{T}$ be a finite state system, $\pi\in \mathit{traces}(\mathcal{T})$ a trace, $\leq^{subset}$ a subset similarity relation, and $\Effect \subseteq (2^\AP)^\omega$ an effect property. Suppose $\Cause$ is a cause of $\Effect$ on $\pi$. Then $C$ is a universal preimage of $E$ under a prefix-continuous, prefix-closed, and prefix-compact function. 
\end{proposition}

The immediate corollary from Theorem~\ref{hierarchy_preservation} and Proposition~\ref{causality_is_up} is the closure of temporal hierarchy under causality.

\begin{corollary}[Causality Closure] \label{causality_preservation}
    Let $\mathcal{T}$ be a finite state system, $\pi\in \mathit{traces}(\mathcal{T})$
    a trace, $\leq^{subset}$ a subset similarity relation, and $\Effect \subseteq (2^\AP)^\omega$ an effect property. Suppose $\Cause$ is a cause of $\Effect$ on $\pi$ in $\TT$. Then the following statements hold.

    \begin{enumerate}
        \item If $E$ is a safety property, then $C$ is a safety property.
        \item If $E$ is a guarantee property, then $C$ is a guarantee property.
        \item If $E$ is a recurrence property, then $C$ is a recurrence property.
    \end{enumerate}
\end{corollary}

Note that we will discuss the classes that are not closed in this way in Section~\ref{sec:non-closure}, after proving Theorem~\ref{hierarchy_preservation} with results from general topology in the next section.

\subsection{Proof via Topology} \label{topological_closure}
This section is dedicated to a topological characterization of the concepts presented before and proves Theorem \ref{hierarchy_preservation} using results from general topology. We first introduce the basic topological definitions. In parallel, we define the Cantor metric and Cantor topology on the space of words.

\subsubsection{Metric spaces}

A metric space is a set equipped with a notion of distance between its elements. Formally, a \emph{metric space} is a pair $(M,d)$, where $M$ is a set and $d: M\times M \to \RR$ is a non-negative, symmetric function, which satisfies the triangle inequality and $d(x,y) = 0 \Longleftrightarrow x=y$. 

The set of words over an alphabet $\Sigma$ can be equipped with a distance function $d_C$ called the \emph{Cantor distance}. The Cantor distance between two words $\pi_1$ and $\pi_2$ is defined as $0$ if they are identical, and as $d_C(\pi_1,\pi_2) := 2^{-j}$ otherwise, where $j$ is the length of the longest common prefix of $\pi_1$ and $\pi_2$. Intuitively, the closer two words are in the Cantor metric, the longer their common prefix. Note that the Cantor distance is bounded by $1$.

For a metric space $M$, a set $O \subseteq M$ is \emph{open} if, for every $x\in O$ there exists $\epsilon>0$, such that the ball $B_\epsilon(x)$ of radius $\epsilon$ with the center in $x$ lies entirely in $O$.
$$B_\epsilon(x) := \{ x'\in M \mid d(x,x') \leq \epsilon\} \subseteq O \enspace.$$

The complement of an open set is called \emph{closed}. The collection of open sets forms a \emph{topology} on $M$, induced by the metric $d$.

For $\Sigma^\omega$ with the Cantor metric, open sets are of the form $X\Sigma^\omega$, where $X\subseteq \Sigma^*$ is a finite word language according to Proposition 3.1 in \cite{infinite-words}. Intuitively, a set $L\subseteq \Sigma^\omega$ is open if, for every $\pi \in L$, there exists $n\in \NN$, such that any $\pi'$ coinciding with $\pi$ in the first $n$ letters also belongs to $L$. This is equivalent to saying that the ball of radius $2^{-n}$ centered at $\pi$ is contained in $L$. 
The topology induced by the Cantor distance is called the \emph{Cantor topology}.

According to Proposition 3.5 in \cite{infinite-words}, a set $L\subseteq \Sigma^\omega$ is closed if for every $\pi\in \Sigma^\omega$: $\pref(\pi) \subseteq \pref(L) \Rightarrow \pi \in L$.

The function $f: M_1 \to M_2$ between two metric spaces $(M_1,d_1)$ and $(M_2,d_2)$ is \emph{continuous} if for every $x\in M_1$ and every $\epsilon > 0$, there exists $\delta > 0$ such that: 
$$\forall x'\in M_1 : \ \text{if} \ d_1(x,x') < \delta, \ \text{then} \ d_2(f(x),f(x'))< \epsilon \enspace .$$

Cantor topology derives its name from the fact that the space of words with this topology can be continuously injected into the Cantor space, as shown by Plotkin \cite{Plotkin1976APC}.

\subsubsection{Borel hierarchy}

In a metric space $M$ the union of any collection of open sets is also open and the intersection of a finite number of open sets remains open. However, the intersection of a countable collection of open sets may no longer be open. The set of countable intersections of open sets is noted as $\boldsymbol{\Pi_2}$.

Symmetrically, the set of closed sets is closed under arbitrary intersections but only finite unions. The set of countable unions of closed sets is denoted as $\boldsymbol{\Sigma_2}$. Together, $\boldsymbol{\Sigma_2}$ and $\boldsymbol{\Pi_2}$ form the second level of the Borel hierarchy. 

The Borel hierarchy organizes subsets of the metric space $M$ into classes. A set is called Borel if it belongs to some level of the Borel hierarchy. The first level of the Borel hierarchy consists of the set of open sets $\boldsymbol{\Sigma_1}$ and the set of closed sets $\boldsymbol{\Pi_1}$. The higher levels are defined recursively as follows:
$$\boldsymbol{\Sigma_n} := \{ \cup_{i\in \NN} X_i \mid X_i \in \boldsymbol{\Pi_{n-1}} \},$$
$$\boldsymbol{\Pi_n} := \{X \mid M\setminus X \in \boldsymbol{\Sigma_{n}}\},$$
$$\boldsymbol{\Delta_n} := \boldsymbol{\Sigma_n} \cap \boldsymbol{\Pi_n}.$$
The first two and a half levels of the Borel hierarchy for the set of infinite words $\Sigma^\omega$ with the Cantor topology correspond to the temporal hierarchy, as established by Mana and Pnueli~\cite{MannaP89}.

\begin{theorem} [Mana and Pnueli~\cite{MannaP89}] \label{hierarchy_correspondence}
    Let $L \subseteq \Sigma^\omega$ be an infinite word language.

    \begin{enumerate}
        \item $L$ is a safety property iff $L$ is a closed set.
        \item $L$ is a guarantee property iff $L$ is an open set.
        \item $L$ is a obligation property iff $L \in \boldsymbol{\Delta_2}$.
        \item $L$ is a recurrence property iff $L \in \boldsymbol{\Pi_2}$.
        \item $L$ is a persistence property iff $L  \in \boldsymbol{\Sigma_2}$.
        \item $L$ is a reactivity property iff $L \in \boldsymbol{\Delta_3}$.
    \end{enumerate}
\end{theorem}

Thus, the problem of the closure of temporal classes is equivalent to the problem of the closure of Borel classes.

A fundamental fact about Borel classes is that they are closed under the continuous preimage of a function \cite{infinite-words}.

\begin{proposition} [\cite{infinite-words}] \label{continuous_borel}
    The preimage of a Borel set under a continuous function is a Borel set of the same Borel class.
\end{proposition}

\subsubsection{Hausdorff metric}

A set $X\subseteq M$ is \emph{compact} if for every collection of open sets $\{ Y_i\}_{i\in \alpha}$ that covers $X$: $X\subseteq \bigcup_{i\in \alpha }Y_i$ there exists a finite subset of $\{ Y_i\}_{i\in \alpha}$ that also covers $X$. Every compact subset of a metric space is closed.

According to Proposition 3.12 in~\cite{infinite-words}, a set $L \subseteq \Sigma^\omega$ is compact in the Cantor topology if it is closed and for every $n\in \NN$ the set $\pref_n(L)$ is finite.

Since we are addressing the problem of closure of Borel classes under the preimages of a function $f:\Sigma_1^\omega \to \mathcal{P} (\Sigma_2^\omega )$, which maps traces to sets of traces, it is essential to define a metric on the set of sets of traces. If $f$ is prefix-closed and prefix-compact it maps traces to compact sets of traces.  Therefore, it suffices to define a metric on the set of compact sets. For a set $X \subseteq M$, we denote the set of all nonempty compact subsets of $X$ as $\KK(X)$. The Hausdorff distance between two nonempty sets $X, Y \subseteq M$ is defined as follows:
$$d_H (X,Y) = max (\underset{x\in X}{sup} \ d(x,Y), \underset{y\in Y}{sup} \ d(X,y)) \enspace,$$
$$\text{where }d(a,B) = \underset{b\in B}{\mathit{inf}} \ d(a,b)\text{ for }B\subseteq M \text{ and }a\in M.$$

On the set $\KK(M)$, the Hausdorff distance $d_H$ is a metric, which induces a topology known as the Hausdorff topology~\cite{Goldlcke2014VariationalA}. The next proposition shows that prefix-continuous functions are essentially functions that are continuous in Hausdorff topology. 

\begin{proposition} \label{prefix_proposition}
    A function $f: \Sigma_1^\omega \to \mathcal{P}(\Sigma_2^\omega)$ is prefix-continuous, prefix-closed and prefix-compact if and only if it is a continuous function from the metric space $\Sigma_1^\omega$ to the metric space $\KK(\Sigma_2^\omega)$ with the Hausdorff metric $d_H$.
\end{proposition}

\subsubsection{Borel Classes in Hausdorff topology}

Combining Propositions \ref{continuous_borel} and \ref{prefix_proposition} we find that prefix-closed, prefix-compact, and prefix-continuous functions preserve Borel classes of the Hausdorff topology. Therefore, it remains to determine how Borel classes behave when lifted from $\Sigma^\omega$ to $\KK (\Sigma^\omega)$.

\begin{lemma} \label{Hausdorff_lemma}
    Let $M$ be a metric space with a subset $X\subseteq M$. Then the following statements are true.
    \begin{enumerate}
        \item If $X$ is open in $M$, then $\KK(X)$ is open in $\KK(M)$.
        \item If $X$ is closed in $M$ then $\KK(X)$ is closed in $\KK(M)$.
        \item If $X\in \boldsymbol{\Pi_2} $ in $M$ then $\KK(X)\in \boldsymbol{\Pi_2} $ in $\KK(M)$.
    \end{enumerate}
\end{lemma}

\begin{proof}
    The first two statements can be demonstrated through the equivalence of the Hausdorff topology and the Vietoris topology on the space of compact subsets \cite{Michael1951TopologiesOS}. However, for the sake of clarity, we present an explicit proof.

    \emph{1)}: Let $X\subseteq M$ be an open set and let $K \in \KK(X)$ be a compact subset of $X$. Define a continuous function $f$ on $K$ for every $x\in K$ as follows: $f(x) := d(x,M\setminus X)$.
    This is a continuous function from the compact set to $\RR$. By the Extreme Value Theorem \cite{Rudin1964PrinciplesOM}, $f$ is bounded, and there exists $q\in K$ such that $f(q) = \mathit{inf}_{x\in K} f(x)$. Since $q\notin M\setminus X$ and $M\setminus X$ is closed, we conclude that $f(q) = d(q,M\setminus X)>0$. 

    By the definition of Hausdorff distance, for every set $Z \in \KK(M)$ with $d_H(K,Z) < f(q)$, it follows that for every $x\in Z$, $d(x, K) < f(q)$. Consequently, $x\notin M\setminus X$, thus, $Z\subseteq X$. Therefore the ball $B_{f(q)}(K) $, of radius $f(q)$ centered at $K$ lies entirely within $\KK(X)$. Since this holds for any $K\in \KK(X)$, we have shown that $\KK(X)$ is open.

    \emph{2)}: Let $X\subseteq M$ be a closed set and let $K\in \KK(M)$ a compact set not in $\KK(X)$. Then there exists $x\in K$, such that $x\notin X$, and since $X$ is closed, $d(x,X)>0$. 
    
    For any $Z\in \KK(M)$ with $d_H(Z,K)< d(x,X)/2$, there exists $z\in Z$, such that $d(x,z)<d(x,X)$. This implies $z\notin X$, thus, $Z\notin \KK(X)$. Therefore, the ball $B_{d(x,X)/2}(K) $, of radius $d(x,X)/2$ centered at $K$, lies entirely within $\KK(M) \setminus \KK(X)$. Since this holds for any $K \in \KK(M) \setminus \KK(X)$, we have shown that $\KK(M) \setminus \KK(X)$ is open. Thus, $\KK(X)$ is closed.

    \emph{3)}: Let $X$ be a set from the Borel class $\boldsymbol{\Pi_2 }$ in $M$. Then by definition of $\boldsymbol{\Pi_2 }$:
    $$X = \bigcap_{n\in \NN} Y_n \enspace,$$
    where each $Y_n$ is an open set. Thus,
    $$\KK(X) = \bigcap_{n\in \NN}\KK(Y_n) \enspace.$$
    Each $\KK(Y_n)$ is open in $\KK(M)$ by the first statement, completing the proof.
\end{proof}

Finally, we combine all the results to prove Theorem~\ref{hierarchy_preservation}.

\begin{proof}[Proof of Theorem~\ref{hierarchy_preservation}]
    By Proposition \ref{prefix_proposition}, $f$ can be viewed as a continuous function between two metric spaces $f: \Sigma_1^\omega \to \KK(\Sigma_2^\omega)$. 
    By the definition of the universal preimage, we get
    $$f^{-1}_U (L) = f^{-1}(\KK(L)) \enspace.$$
    By Lemma \ref{Hausdorff_lemma}, we know that $\KK$ preserves open, closed, and $\boldsymbol{\Pi_2 }$ subsets. By Proposition \ref{continuous_borel} we know that $f^{-1}$ preserves all Borel classes, as $f$ is continuous. 
    Thus, we prove all three statements of the theorem, since by Theorem \ref{hierarchy_correspondence} safety, guarantee, and recurrence properties are exactly open, closed, and $\boldsymbol{\Pi_2 }$ sets.
\end{proof}

\subsection{Existential preimage and non-closure of persistence}\label{sec:non-closure}

Since the universal preimage formalizes the universal quantification inherent in temporal causes, it is natural to seek a similar formalization for existential quantification. To this end, we introduce the concept of the existential preimage.

\begin{definition}
Let $X$ and $Y$ be sets, and let $f: X \to \mathcal{P}(Y)$ be a function.
  The \textbf{existential preimage} of the set $S$ under $f$, denoted $f^{-1}_E(S)$, is the set of elements whose images under $f$ intersect $S$: $f^{-1}_E(S) := \{x\in X \mid f(x) \cap S \neq \emptyset\}$. 
\end{definition}

The existential preimage is dual to the universal preimage in the sense that, for a set $S\subseteq Y$, $f^{-1}_E(S) = X \setminus (f^{-1}_U Y\setminus S))$. Using this duality we immediately get the following corollary from Theorem \ref{hierarchy_preservation}.

\begin{corollary}
    Assume a function $f:\Sigma_1^\omega \to \mathcal{P}(\Sigma_2^\omega)$ is prefix-continuous, prefix-closed, and prefix-compact. Then the safety, guarantee, and persistence classes are closed under $f^{-1}_E$.
\end{corollary}

Hence, the persistence class is closed under the existential preimage, but as we show in the following, it is not closed under the dual universal preimage which encodes causal inference. The rest of the subsection discusses its behavior under the universal preimage, along with the obligation class which also is not closed under the universal preimage.

First, we consider causal inference with the subset similarity relation $\leq^{subset}$, a special case of the universal preimage by Proposition \ref{causality_is_up}. We prove that an obligation effect can have a non-obligation cause. 

\begin{theorem}\label{nclosure_obligation}
    There exists a system $\TT$, a trace $\pi$, and an obligation $\omega$-regular effect $E$, such that the cause of $E$ on $\pi$ in $\TT$ with similarity relation $\leq^{subset}$ is not an obligation property.
\end{theorem}

\begin{proof}
    Define the input alphabet $I := \{ a \}$. The output alphabet is empty, i.e., $O:= \emptyset$. The system is the trivial set of all possible traces: $\traces(\TT):= (2^{I\cup O})^\omega.$ The observed trace enables $a$ continuously: $\pi := a^\omega$. Consider the effect $E:= (\LTLsquare a) \vee (\LTLdiamond (\neg a\wedge (\LTLnext a))).$ $E$ is an obligation property since it is a disjunction of safety and guarantee properties. In essence, $E$ does allow any trace but $a^+ (\neg a)^\omega$. The cause of $E$ on $\pi$ in $\TT$ is $\LTLsquare \LTLdiamond a$, since if $\pi''$ satisfies $\LTLsquare \LTLdiamond a$, then clearly any $\pi'\leq_\pi^{subset} \pi''$ satisfies $\LTLsquare \LTLdiamond a$, and hence $E$. If $\pi''$ does not satisfy $\LTLsquare \LTLdiamond a$ that $\pi''\in (2^I)^n (\neg a)^\omega$ for some $n$, hence $a^n (\neg a)^\omega \leq _\pi^{subset} \pi''$, thus $\pi''$ is not in the cause.
\end{proof}

Together with Theorem~\ref{hierarchy_preservation} that states that recurrence is closed, this immediately implies that persistence also is not closed under causality. This is because from Theorem~\ref{hierarchy_preservation} it follows that cause for the obligation property is a recurrence property, so the only way it cannot be an obligation property is by not being a persistence property.

\begin{corollary}\label{nclosure_persistence}
    There exists a system $\TT$, a trace $\pi$, and a persistence $\omega$-regular effect $E$, such that the cause of $E$ on $\pi$ in $\TT$ with similarity relation $\leq^{subset}$ is not a persistence property.
\end{corollary}

Besides the fact that persistence is not closed under causality and consequently under the universal preimage, we want to investigate for how many Borel classes we can find similar counterexamples. Clearly, we must go beyond the $\omega$-regular setting, as $\omega$-regular properties are contained in $\boldsymbol{\Delta_3}$ and are themselves closed under causal inference. 

 Existential and universal preimages are dual to each other, as recurrence ($\boldsymbol{\Pi_2}$) and persistence ($\boldsymbol{\Sigma_2}$) classes. Thus, it suffices to examine the behavior of recurrence properties under existential preimage.

\begin{lemma}\label{omega_power}
    Let $A\subseteq \Sigma^*$ be a finite word language. Then the infinite word language $A^\omega$ is an existential preimage of a recurrence property under a prefix-continuous, prefix-closed, and prefix-compact function.
\end{lemma}

\begin{proof}
    Define the alphabet $\Sigma'$ as consisting of primed copies of symbols from $\Sigma$. 

    The language $A'\subseteq (\Sigma \cup \Sigma')^\omega$ consists of words from $A$ in which the last letter is replaced with its primed version. Formally:
    $$A' := \{ \pi_{(|\pi|-1)} \pi_{|\pi|}' \mid \pi \in A  \} \enspace.$$

    Clearly, $(A')^\omega$ is a recurrence property, as a word $\pi$ is in $(A')^\omega$ if and only if it has infinitely many prefixes from $(A')^* $. 

    Define a prefix-continuous, prefix-closed, and prefix-compact function $f: \Sigma^\omega \to (\Sigma\cup \Sigma')^\omega$ for $\pi\in \Sigma^\omega$ as follows:
    $$f(\pi):= \{ \pi' \mid \forall i: \ \pi(i) = \pi'(i) \text{ or } \pi(i)' = \pi'(i)  \} \enspace.$$

    Essentially, $f$ randomly replaces each symbol in $\pi$ with its primed version. Intuitively it tries to split $\pi$ into words from $A$, with primed letters representing the endings of words from $A$. If $f$ can split $\pi$ into the words from $A$, this splitting is in $f(\pi)\cap (A')^\omega$. Hence, $A^\omega = f^{-1}_E((A')^\omega)$.
\end{proof}

\begin{theorem}
    There exists a recurrence property $L$ and a prefix-continuous, prefix-closed, and prefix-compact function $f$, such that $f_E^{-1}(L)$ is not a Borel set. 
\end{theorem}

\begin{proof}
    There exists a finite language $A$, such that $A^\omega$ is not Borel, as shown in \cite{Finkel2003BorelHA}. Therefore, the theorem follows directly from Lemma \ref{omega_power}.
\end{proof}

\section{Complexity}\label{sec:complexity}

In this section, we take a closer look at the complexity of synthesizing causes as automata and study bounds for the size of these automata. We show lower bounds for all property classes, which are the first lower bounds for the problem and witness that the exponential scaling in the algorithm of Finkbeiner et al.~\cite{FinkbeinerFMS24} cannot be avoided. However, we also show that the upper bounds can still be improved for several property classes as logarithmic factors can be avoided. We focus our attention to the case of subset similarity relation $\leq^{subset}$. The presented characterization is precise with respect to the system size, but with respect to the effect size, there remains a minor gap between the lower and upper bounds on the higher levels of the hierarchy.

We heavily use the characterization of causes as downward closed sets of traces satisfying the effect (cf.\ Lemma~\ref{lem:cause}).
For the upper bounds, we present the algorithms that output the set from Lemma \ref{lem:cause} matching a cause when it exists.

\subsection{Lower bounds}

First, we prove that the NBW and NCW complementation problems can be reduced to the cause synthesis problem linearly in the size of the system with persistence or recurrence effects, respectively.

The complement of an NBW can be expressed as a UCW (universal co-Büchi automaton) with the same structure. Hence it can be viewed as a universal quantification over a DCW automaton. Similarly, the complement of an NCW can be viewed as a universal quantification over a DBW automaton.

The trick in the proof is now to interpret the automaton to be complemented as the system in a causal inference tasks. The acceptance condition of the complement automaton can be encoded as a recurrence/persistence property for Büchi and Co-Büchi acceptance, respectively. It remains to ensure that only identical words are related in the similarity relation, such that the universal quantification as described in Lemma~\ref{lem:cause} only ranges over the same word. In the end, a word is then in the cause automaton only if all runs satisfy the complementary acceptance condition. Hence, the language of the cause automaton is exactly the complement of the original automaton language. Technical details of this construction are provided in the proof of the following lemma.

\begin{lemma}\label{NBW_NCW_lemma}
    For every NBW (NCW) $A$ there exists a system $\TT$ with $|A|+1$ states, an effect $E$ represented by a DCW (DBW) of constant size and a trace $\pi$ of constant size, such that if the cause of $E$ on $\pi$ in $\TT$ is expressed as a NBW $C$ then there exists a NBW for the complement of $L(A)$ of the size $\OO(|C|)$.
\end{lemma}

\begin{proof}

Denote  $A = ( Q, \Sigma, q_0, F, \Delta )$, with $\Sigma = \{1,\dots,m\}$ and $|Q| = n$. Let us define a system $\TT= (S,s_0, AP, \delta, l)$, where $AP = I \bigcup O$.
$$I := \{ i_1,\dots , i_k, j_1,\dots, j_{\lceil \log( n) \rceil} \}.$$
Here $k$ is a minimal integer such that $\binom{k}{k/2} \geq m$. Outputs and states are defined as follows.

$$O := \{ o\}, \enspace S := Q \cup \{ s_{\top}\}, s_0 := q_0.$$
To define $\delta$ at first we need to encode pairs $q,\sigma \in Q\times \Sigma$ as elements from $2^I$. We fix an injective function $\mathit{enc}_Q$ from $Q$ to the set of subsets of $\{ j_1,\dots,j_{\lceil \log( n) \rceil} \}$. 
$$\mathit{enc}_Q: Q \to \mathcal{P}(j_1,\dots,j_{\lceil \log( n) \rceil}).$$
Additionally, we fix an injective function $\mathit{enc}_\Sigma$ from $\Sigma$ to the set of subsets of $\{ i_1,\dots , i_k \}$ of the size $k/2$. The number of such subsets is $\binom{k}{k/2} \geq m$ by choice of $k$, hence such injective function exists.
$$\mathit{enc}_\Sigma: \Sigma \to \{ V\subseteq \{ i_1,\dots , i_k \} \mid  \ |V| = k/2\}$$
Please note that for every two different $\sigma,\sigma' \in \Sigma$ subsets $\mathit{enc}_\Sigma(\sigma)$ and $\mathit{enc}_\sigma(\sigma')$ are incomparable.
\begin{equation}
\forall \sigma,\sigma' \in \Sigma: \sigma\neq \sigma' \Rightarrow \mathit{enc}_\Sigma(\sigma) \not\subseteq \mathit{enc}_\Sigma(\sigma'). \tag*{(*)} \label{incomparability}
\end{equation}
We define the transition function $\delta$ as follows.
$$\delta(q,\mathit{enc}_Q(q') \cup \mathit{enc}_\Sigma(\sigma)) := \begin{cases}
    q', \ \ \ \ \text{if } q'\in \Delta(q,\sigma), \\
    s_\top, \ \ \  \text{otherwise.}
\end{cases}$$
For every $q\in Q$ and $W\subseteq I$ which cannot be presented as $\mathit{enc}_Q(q') \cup \mathit{enc}_\Sigma(\sigma)$ for $q'\in Q$ and $\sigma\in \Sigma$:

$$\delta(q,W) := s_\top \enspace\text{ and }\enspace \delta(s_\top, W) := s_\top \enspace .$$
If $A$ is an NBW labeling $l_{\mathit{NBW}}$ is defined as follows:
$$l_{\mathit{NBW}}(s) := \begin{cases}
    \{o\} \ \ \ \ \text{if } s\in F \enspace ,\\
    \emptyset \ \ \ \ \ \ \ \ \text{if } s\in Q\setminus F \cup \{ s_\top\} \enspace .
\end{cases}$$
If $A$ is an NCW the labeling $l_{\mathit{NCW}}$ is defined differently:
$$l_{\mathit{NCW}}(s) := \begin{cases}
    \{o\} \ \ \ \ \text{if } s\in Q\setminus F \enspace ,\\
    \emptyset \ \ \ \ \ \ \ \ \text{if } s\in F \cup \{ s_\top\}\enspace .
\end{cases}$$
Trace $\pi = \emptyset^\omega$ does not depend on the automaton. We define the effect depending on whether $A$ is an NBW or an NCW:
$$E_{\mathit{NBW}} := \LTLdiamond \LTLsquare \neg o \enspace\text{ or }\enspace E_{\mathit{NCW}} := \LTLsquare \LTLdiamond \neg o \enspace .$$
For a word $\sigma \in \Sigma^\omega$ we denote the encoding of $\sigma$ with $2^I$ as $\mathit{Enc}_I(\sigma)\in (2^I)^\omega$. For every $k$ the $k$-th letter of $\mathit{Enc}_I(\sigma)$ is defined as follows:
$$\mathit{Enc}_I(\sigma)_k = \mathit{enc}_\Sigma(\sigma_k) \cup \{j_1, \dots, j_{\lceil log( n) \rceil} \} \enspace .$$

\noindent\emph{Claim:} For every word $\sigma \in \Sigma^\omega$: $\sigma\in \overline{A} \Longleftrightarrow \mathit{Enc}_I(\sigma)\in C.$

\medskip
\emph{The first direction:} $\sigma\in \overline{A} \Longrightarrow \mathit{Enc}_I(\sigma)\in C.$ Assume a word $\sigma \in \overline{A}$. All runs of $A$ on $\sigma$ must be rejecting, hence they must visit $F$ finitely many times if $A$ is an NBW or infinitely many times if $A$ is an NCW.

Assume a trace $\pi' \in (2^{I\cup O})^\omega$, such that $\pi'\leq^{subset}_\pi \mathit{Enc}_I(\sigma)$. If while producing $\pi'$ the system $\TT$ visits state $s_\top$, then $\pi' $ satisfies $E$. 
Suppose while producing $\pi'$ system  $\TT$ does not visit $s_\top$. Thus, by the  definition of the transition system for every $k$: $\pi'_k (I) = \mathit{enc}_Q(q_k)\cup \mathit{enc}_\Sigma(\sigma_k)$, where $\{q_j\}_{j\in \NN}$ is a run of $A$ on $\sigma$. Please note that the $\Sigma$ part of $\pi'$ in this case is the encoding of $\sigma$ and not an encoding of any other trace from $\Sigma^\omega$, by the fact that $\pi' \leq^{subset}_\pi \mathit{Enc}_I(\sigma)$ and the property of encoding \ref{incomparability}.

As we noted before, $\{q_j\}_{j\in \NN}$ must be rejecting, hence it cannot visit $F$ infinitely many times in the case of NBW or it must visit $F$ infinitely many times in the case of NCW. Thus $\pi'$ satisfies $E$. We proved that $\mathit{Enc}_I(\sigma) \in C$, since for any $\pi' \leq^{subset}_\pi \mathit{Enc}_I(\sigma)$: $\pi'\in E$.

\medskip
\emph{The second direction:} $\sigma\in \overline{A} \Longleftarrow \mathit{Enc}_I(\sigma)\in C.$ Assume a word $\sigma \in \Sigma^\omega$ such that $\mathit{Enc}_I(\sigma)\in C$. Let us take a run $\{q_k\}_{k\in\NN}$ of $A$ on $\sigma$. Let us define the word $\pi' \in (2^I)^\omega$ as follows. For every $k$:
$$\pi'_k := \mathit{enc}_Q(q_k) \cup \mathit{enc}_\Sigma (\sigma_k) \enspace .$$
By the definition of $\leq^{subset}$ we get $\pi'\leq^{subset}_\pi \mathit{Enc}_I(\sigma)$. Hence, $\TT(\pi') $ satisfies $E$ by the definition of the cause, since $\mathit{Enc}_I(\sigma)\in C$. Hence, by the definition of $E$ $\{q_k\}_{k\in\NN}$ does not visit $F$ infinitely often in the case of NBW or visits $F$ infinitely often in the case of NCW, thus it is a rejecting run of $A$. Therefore, $\sigma \in \overline{A}$ since every run of $A$ on $\sigma$ is rejecting.
\end{proof}

Similarly, we provide the linear encoding of the NFW complementation problem into the cause synthesis problem with safety or reachability effect.

\begin{lemma} \label{NFW_lemma}
    For every NFW $A$ there exists a system $\TT$ with $|A|+1$ states, an effect $E$ in the form of a safety (or 
    reachability) DBW of constant size and a trace $\pi$ of constant size, such that if the cause of $E$ on $\pi$ in $\TT$ expressed as a NBW $C$ then there exists a NFW for the complement of $L(A)$ of the size $\OO(|C|)$.
\end{lemma}

\begin{proof}[Proof]
    Similar to Lemma~\ref{NBW_NCW_lemma}. We add to the alphabet of $A$ one additional symbol $\#$, which represents the end of the word. Then we construct $\TT$,$E$, and $\pi$ in the same way as we did in the Lemma~\ref{NBW_NCW_lemma} to track all possible executions of $A$. The only difference is that now $E$ detects if the last state that appeared before the first occurrence of $\#$ was accepting in $A$ or not.

    Please note that $E$ in this case can be a safety or reachability DBW since we can either accept or reject words without occurrences of $\#$. If $C$ is the NBW for the cause, we can easily turn it into NFW for $\overline{A}$ just by calling a state accepting if $C$ accepts $\#^\omega$ from this state.
\end{proof}

Using the provided reductions of complementation problems we derive the lower bounds of the cause size with respect to the system $\TT$ for different temporal classes from the well-known lower bounds on the automata complementation problems.   

\begin{theorem}\label{lower_system}
The following states lower bounds for the cause automaton with respect to the size of the system.
    \begin{enumerate}
        \item\label{thm51} An NBW for the cause of a persistence effect in a system $\TT$ requires at least $2^{\Omega(|\TT| log |\TT|)}$ states in the worst case.
        
        \item\label{thm52} An NBW for the cause of a recurrence effect in a system $\TT$ requires at least $\Omega(3^{|\TT|})$ states in the worst case.

        \item\label{thm53} An NBW for the cause of a safety or guarantee effect in a system $\TT$ requires at least $\Omega(2^{|\TT|})$ states in the worst case.
    \end{enumerate}
\end{theorem}

\begin{proof} 1) For an NBW with $n$ states, an NBW for the complement requires at least $2^{\Omega(nlogn)}$ states in the worst case~\cite{Vardi2007TheBC}. By Lemma \ref{NBW_NCW_lemma} NBW complementation can be reduced to cause synthesis with persistence effect linearly in the size of $\TT$.
    
    2) For an NCW with $n$ states, an NBW for the complement requires at least $\Omega(3^n)$  states in the worst case~\cite{Boker2010AlternationRI}. By Lemma~\ref{NBW_NCW_lemma}, NCW complementation can be reduced to cause synthesis with a recurrence effect linearly in the size of the system $\TT$.

    3) For an NFW with $n$ states, an NFW for the complement requires at least $\Omega(2^{n})$ states in the worst case~\cite{Hopcroft1979IntroductionTA}. By Lemma~\ref{NFW_lemma}, NFW complementation can be reduced to cause synthesis with a safety or guarantee effect linearly in the size of the system.
\end{proof}

We turn our focus to the complexity with respect to the effect $E$. We prove a doubly exponential lower bound. For that purpose for every $n\in \NN$, we define a language that is the cause of an effect of size $\OO(n)$, while any NBW recognizing this language requires at least $2^{2^{\Omega(n)}}$ states.

For $n\in \NN$ and $\pi \in \Sigma^*$ let us denote as $subword_n(\pi)$ the set of words of length $n$ that appear in $\pi$ from a position divisible by $n$.
$$subword_n(\pi) := \{ \pi_{kn},\pi_{kn+1},\dots, \pi_{(k+1)n-1} \mid k\in \NN \}$$

Let us define a finite word language $L_n \subseteq \{ 0,1,\#\}^*$ for $n\in \NN$ as follows.
$$L_n := \{ \pi \in \{ 0,1,\#\}^* \mid \forall w\in \{ 0,1\}^n , \ w\in subword_n(\pi)\}$$
In other words, $L_n$ consists of words $\pi$ such that $\{ 0,1\}^n \subseteq subword_n(\pi)$. We prove the doubly exponential lower bound on the NBW which recognizes $L_n$.

\begin{lemma}
    An NFW for $L_n$ requires at least $2^{2^{\Omega(n)}}$ states.
\end{lemma}

Now we show that $L_n$ is the cause of an effect of size $\OO(n)$. The idea uses that a word $\pi$ belongs to $L_n$ if $subword_n(\pi)$ contains every word $w\in \{ 0,1\}^n$. 

For any single $w$ consider the infinite word $w^\omega$. For such a word there exists a position $kn$ such that $\pi$ and $w^\omega$ coincide for $n$ consecutive symbols starting at $kn$. This position marks where $w$ appears in $\pi$. 
    
To verify this, an automaton with $\OO(n)$ states suffices.  It tracks the position modulo $n$ and nondeterministically selects $kn$ as the starting point where $\pi$ and $w^\omega$ coincide. Hence, universally quantifying this check over all $w$ we get a construction that recognizes $L_n$.

\begin{lemma}
    For every $n$ there is a system $\TT$ and a trace $\pi_\emptyset$ of constant sizes and a safety (or reachability) effect of size $\OO(n)$, such that if the cause of $E$ on $\pi_\emptyset$ in $\TT$ expressed as a NBW $C$ then there exists a NFW for $L_n$ with $\OO(|C|)$ states.
\end{lemma}

\begin{proof}
    Let us define an input alphabet consisting of four letters $I = \{i_0,i_1,i_\#,i_* \}$ and an output alphabet consisting of one letter $O = \{ o\}$. The system is trivial and models every trace over the alphabet. The trace is the trivial trace without enabled atomic propositions: $\traces(\TT) := (2^{I\cup O})^\omega, \enspace \pi_\emptyset
    := \emptyset^\omega$.
    The effect $E$ is defined as the union $E:= E_1 \cup E_2 \cup E_3.$
    The first part $E_1$ requires that at some position all input variables $i_0,i_1,i_\#, i_*$ become false.
    \begin{align*}
        E_1 := \{\pi \in (2^{I\cup O})^\omega \mid \exists k : &\neg(\pi_k(i_0) \vee \pi_k(i_1)\\ &\vee \pi_k(i_\#) \vee \pi_k(i_*))) \}
    \end{align*}
    $E_2$ requires that the output part of the trace does not repeat every $n$ states, in other words, there is a position such that after $n$ positions the output variable takes a different value. 
    $$E_2 := \{\pi \in (2^{I\cup O})^\omega | \exists k :\pi_k(o) \neq \pi_{k+n}(o)\}$$
    $E_3$ requires that from some position divisible by $n$ before $i_*$ becomes true for the first time variables $i_1$ and $o$ take the same value $n$ steps in the row.
    \begin{align*}
        E_3 := \{&\pi \in (2^{I\cup O})^\omega | \exists k \ \forall l<kn \ \forall j<n: \\
        &\neg \pi_{l}(i_*) \wedge( \pi_{kn+j}(i_1) = \pi_{kn+j}(o)) \}
    \end{align*}
    Since $E$ must remember only the number of the steps modulo $n$ it can be modeled by an NBW with $n$ states which nondeterministically decides on which step $E_2$ or $E_3$ is satisfied.

    For a word $\sigma \in \{0,1,\# \}^*$ we denote the encoding of $\sigma$ with $2^I$ as $\mathit{Enc}_I(\sigma)\in (2^I)^\omega$. For every $k$ the $k$-th letter of $\mathit{Enc}_I(\sigma)$ is defined as follows:
    $$\mathit{Enc}_I(\sigma)_k = \begin{cases}
        \{ i_c \}  & \text{if} \ k<|\sigma| \text{ and } \sigma_k = c \enspace,\\ 
        \{ i_*\} & \text{if} \ k \geq |\sigma| \enspace .
    \end{cases}$$
    
    \noindent\emph{Claim:} $\forall \sigma \in \{ 0,1,\#\}^*$: $\sigma\in L_n \Longleftrightarrow \mathit{Enc}_I(\sigma)\in C.$
    
    Assuming the Claim, given a NBW for the cause $C$ an NFW of size $\OO(|C|)$ for $L_n$ can be easily constructed by d combining $C$ with encoding $\mathit{Enc}_I$ and defining accepting states as the states from which $C$ accepts $\{ i_*\}^\omega$. The rest of the proof is devoted to proving the claim.
    
    \medskip
    \emph{The first direction}: $\sigma\in L_n \Longrightarrow \mathit{Enc}_I(\sigma)\in C$. Let $\sigma$ be a word from $L_n$. Let us prove that $\mathit{Enc}_I(\sigma) \in C$. For that, we need to prove that for every $\pi'\leq^{subset}_{\pi_\emptyset} \mathit{Enc}_I(\sigma)$: $\pi'\in E$. If $\pi'_I \neq \mathit{Enc}_I(\sigma)$, then by the definition of $\leq^{subset}$ and the encoding for some $k$: $$\pi'_k(i_0)=\pi'_k(i_1)=\pi'_k(i_\#)=\pi'_k(i_*) = \bot \text{, hence }\pi' \in E_1.$$ 

    If $\pi'(o)$ does not repeat every $n$ states $\pi'\in E_2$: Assume that $\pi' \notin E_1\cup E_2$. Then the prefix ${\pi'(o)}_{(n)} $ repeats in $\pi'(o)$ every $n$ states. Since $\sigma \in L_n$ the word  ${\pi'(o)}_{(n)}$ considered as a word over $\{0,1 \}$ must appear in $\sigma$ from the position $kn$ for some $k$. Hence, $\pi'(o)$ and $\mathit{Enc}_I(\sigma)(i_1)$ take the same value $n$ steps in the row from position $kn$. Since $\pi' \notin E_1$, we can conclude that $\mathit{Enc}_I(\sigma(i_1)) = \pi'(i_1)$. Thus, $\pi'\in E_3$.

    \medskip
    \emph{The second direction}: $\sigma\in L_n \Longleftarrow \mathit{Enc}_I(\sigma)\in C$. Assume that $\mathit{Enc}_I(\sigma) \in C$. Let us prove that $\sigma\in L_n$. For that, we need to prove that for every $w\in \{ 0,1\}^n$: $w\in subword_n(\pi)$.   

    Consider $\pi' \in (2^{I\cup O})^\omega $ such that $\pi'_I = \mathit{Enc}_I(\sigma)$ and for every $k$: $\pi'_k(o) = w_{k \% n}$, where $k\%n$ is $k$ modulo $n$. Obviously $\pi'\leq^{subset}_{\pi_\emptyset} \mathit{Enc}_I(\sigma)$, hence $\pi'\in E$. Moreover, $\pi'\in E_3$, since $\pi'\notin E_1\cup E_2$.

    Hence, there exists $k$, such that $\pi'(o)$ and $\pi'(i_1)$ take the same value $n$ steps in a row from the step $kn$. Thus, $w$ appears in $\pi$ from the position $kn$.
\end{proof}

The lower bound in the size of the effect then immediately follows from the last two lemmas.

\begin{theorem}\label{lower_effect}
    An NBW for the cause of a safety or reachability effect, given as an NBW $E$ requires at least $2^{2^{\Omega(|E|)}}$ states in the worst case.
\end{theorem}

    \subsection{Upper bounds}
This subsection establishes upper bounds that match the lower bounds presented in the previous subsection. With respect to the system size $|\TT|$ all the bounds are tight.

%The algorithms presented in this section output the set from Lemma \ref{lem:cause} matching the cause when it exists.

The first upper bounds presented in Theorem \ref{upper_NCW} is $\OO(|\pi| \cdot 3^{(|\TT| \cdot |E| )})$ for the recurrence effect. Note that this assumes that $E$ is provided as a DBW. Converting a recurrence NBW to a DBW requires a blow-up of $2^{\Omega(n \log n)}$ \cite{Colcombet2009ATL}. Hence, the combined upper bound for the cause synthesis for the recurrence effect presented as NBW becomes $2^{2^{\OO(|E| \log|E|)}}$.

In contrast, the lower bound established in Theorem \ref{lower_effect} is $2^{2^{\Omega(|E|)}}$, leaving a question about the tight bound unresolved. The same applies to the persistence class, as the upper bound on the effect derived from \cite{FinkbeinerFMS24} is also $2^{2^{\OO(|E| \log|E|)}}$.

\begin{theorem} \label{upper_NCW}
For a system $\TT$, a similarity relation $\leq^{\mathit{subset}}$, a trace $\pi$ and an effect $E$ given as a DBW, there exists a DBW for the cause of the size $\OO(|\pi| \cdot 3^{(|\TT| \cdot |E| )})$.
\end{theorem}
    
\begin{proof} [Sketch of Proof]
    We construct a universal Büchi automaton of the size $|\TT| \cdot |E|$ for the cause. Using the fact that a universal Büchi automaton of the size $n$ can be translated to a non-deterministic one of the size $3^n$ \cite{Miyano1984AlternatingFA}, we derive the stated upper bound.
    For the construction details of the UBW and the handling of $\pi$ please see the full version of the proof.
\end{proof}

Next, we present upper bounds for safety (Theorem \ref{upper_safety}) and guarantee (Theorem \ref{upper_guarantee}) properties. Both results assume the effect is given as a DFW for the good or bad prefixes respectively. Since translating safety (or guarantee) NBW to the DFW for good (or bad) prefixes requires an exponential blowup \cite{Kupferman1999ModelCO}, these upper bounds match the lower bound established in Theorem \ref{lower_effect}, with respect to the size of the effect given as an NBW.

\begin{theorem} \label{upper_safety}
    For a system $\TT$, a similarity relation $\leq^{\mathit{subset}}$, a trace $\pi$ and a safety effect $E$ given as a DFW $E_{\mathit{bad\_pref}}$ for the bad prefixes of $E$, there exists a DFW for the bad prefixes of the cause of the size $\OO(|\pi| \cdot 2^{(|\TT| \cdot |E_{\mathit{bad\_pref}}| )})$.
\end{theorem}
    
\begin{proof} [Sketch of Proof]
    We construct a NFW which recognizes pairs of finite traces between which a bad prefix exists. The nondeterministic choice represents the selection of such traces. Subsequently,  we convert it to a DFW and combine it with a trace $\pi$. For details, please refer to the full proof.
\end{proof}

\begin{theorem}\label{upper_guarantee}
    For a system $\TT$, a similarity relation $\leq^{subset}$, a trace $\pi$ and a guarantee effect $E$ given as a DFW $E_{\mathit{good\_pref}}$ for the good prefixes of $E$, there exists a DFW for the good prefixes of the cause of the size $\OO(|\pi| \cdot 2^{(|\TT| \cdot |E_{\mathit{good\_pref}}| )})$.
\end{theorem}

\begin{proof} [Sketch of Proof]
    By the definition a finite word $w$ is a good prefix of the cause $C$ if every continuation of it is in $C$. And for that we need that all finite words between $w$ and $\pi$ have a good prefix of $E$. We construct a UFW of size $|\TT| \cdot |E_{good \_ pref}|$ that recognizes the good prefixes of the cause and the convert it into a DFW. For the details on this construction please refer to the full version of the proof.
\end{proof}

\section{Related Work}\label{sec:rel_work}

\paragraph*{Methodology} The topological concepts used in Section~\ref{closure_section} have long been established in mathematics. The Hausdorff distance was introduced by Felix Hausdorff in~\cite{Hausdorff1914GrundzgeDM}. 
The Vietoris topology~\cite{Vietoris} is another topology on the space of subsets. It coincides with the Hausdorff topology on compact subsets but differs when generalized to arbitrary closed subsets~\cite{Michael1951TopologiesOS}. 
These concepts are employed in the Powerdomain theory, studied by Plotkin~\cite{Plotkin1976APC, Plotkin1982APF} and Smyth~\cite{Smyth1978PowerD}. The relation between Powerdomains and the Vietoris topology is discussed in~\cite{Smyth1983PowerDA}. The Vietoris topology can be split into the upper and lower Vietoris topologies. The lower Vietoris topology was applied by Clarkson and Schneider in ~\cite{ClarksonS10} to characterize different classes of hyperproperties.
The established connection between these topological concepts, universal preimages, and causality enable us to use the existing mathematical theory to study the framework.

\paragraph*{Causality} The complexity of checking and computing actual causes in finite, so-called \emph{structural equation models}~\cite{HalpernP01} has been studied extensively~\cite{EiterL02,EiterL06,Halpern15}. In that framework, causes are essentially finite sets of explicit events (comparable to $\boldsymbol{\Delta_0}$ in Figure~\ref{fig:hierarchy_results}) and not arbitrary properties as considered in this paper. Moreover, the structural equation approach does not model time explicitly and can only be used to model system executions up to a fixed bound~\cite{CoenenDFFHHMS22}. An extension to a state-based attribution method for transitions systems has been studied recently~\cite{MascleBFJK21}. For approaches that combine counterfactual reasoning with temporal properties, there are a number of related complexity results not pertaining to temporal causality as considered in this paper. \emph{Event Order Logic}~\cite{Leitner-FischerL13} is an approach that expresses what order of events is causal for some property violation. Upper bounds for the time needed to compute causes in this logic are known to be exponential in the system size~\cite{Leitner-Fischer15}. Parreuax et al.~\cite{ParreauxPB23} define causes for reachability and safety effects in transition systems and two-player games and show that causes can be checked in polynomial time, but they do not consider the problem of cause synthesis or more complex effects. To the best of our knowledge, our paper is the first to establish explicit results on the closure under causal inference. We believe this is because temporal causality is the first formalism in this domain that uses the same language for both cause and effect, in the spirit of earlier work on counterfactual modal logic~\cite{Lewis73a,Stalnaker81}.

\section{Summary \& Conclusion}\label{sec:summary}

We have conducted a detailed investigation of closure under causal inference and complexity of cause synthesis for properties belonging to all classes of the hierarchy of temporal properties~\cite{MannaP89}. Our discoveries can be summarized as follows:
\begin{enumerate}
    \item\label{res:closure} Reachability, guarantee, and recurrence properties are closed under causal inference: An effect from these classes always has a cause from the same class. This complements previous results on $\omega$-regular properties~\cite{FinkbeinerFMS24} and the intersection of safety and guarantee properties~\cite{BeutnerFFS23}.
    \item\label{res:nclosure} Obligation and recurrence properties are not closed in the same way, which completes the picture regarding the hierarchy of $\omega$-regular properties.
    \item\label{res:upper} Based on \ref{res:closure}), we provide improved upper bounds for the size of causes synthesized from reachability and guarantee properties.
    \item\label{res:lower} We show lower bounds on the size of causes for all classes from the hierarchy, which confirm that the known algorithms are optimal in the number of exponents and with respect to the system size. For some classes, a gap in the logarithmic factors with respect to the effect remains.
\end{enumerate}
Contribution~\ref{res:upper}) is of high practical relevance for explaining model-checking results of safety and guarantee properties, which are common in verification tasks. Contributions~\ref{res:closure} and~\ref{res:nclosure} were proven via results for a more abstract mathematical operation that promises to generalize beyond cause synthesis to other problems that involve trace-quantifier alternations. This includes synthesis with incomplete information~\cite{Kupferman2000,KupfermanV01} and automata-based algorithms for hyperproperties~\cite{FinkbeinerRS15,BeutnerCFHK22}. We plan on investigating these connections in future work.

\section*{Acknowledgements}

This work was partially supported by the DFG in project 389792660 (TRR 248 -- CPEC) and by the ERC Grant HYPER (No. 101055412). Funded by the European Union. Views and opinions expressed are however those of the authors only and do not necessarily reflect those of the European Union or the European Research Council Executive Agency. Neither the European Union nor the granting authority can be held responsible for them.

\bibliographystyle{IEEEtranS}
\bibliography{bibliography}

\appendix

\subsection{Detailed Proofs}

\setcounter{theorem}{6} 
\setcounter{lemma}{5} 

\begin{lemma}
    An NFW for $L_n$ requires at least $2^{2^{\Omega(n)}}$ states.
\end{lemma}

\begin{proof}
    Assume $A$ is an NFW that recognizes language $L_n$. Let $Y$ be a set of sets of binary words of length $n$ of size $2^{n-1}$: 
    $$Y = \{ S \subseteq \{ 0,1\}^n\mid   |S| = 2^{n-1}\} \enspace .$$
    For each $y\in Y$ after reading a finite word $\pi$ with $subword_n(\pi) = y$ the NFW $A$ must reach at least one such state from which it accepts a word $\overline{\pi}$ with $subword_n(\overline{\pi}) = \{0,1\}^n \setminus y$. Such states must be different for different $y$, hence $|A| \geq |Y| = \binom{2^n}{2^{n-1}} = 2^{2^{\Omega(n)}}$. 
\end{proof}

\begin{theorem}
For a system $\TT$, a similarity relation $\leq^{\mathit{subset}}$, a trace $\pi$ and an effect $E$ given as a DBW, there exists a DBW for the cause of the size $\OO(|\pi| \cdot 3^{(|\TT| \cdot |E| )})$.
\end{theorem}

\begin{proof}
    We construct a UBW $U$ of the size $|\TT| \cdot |E|$, such that $U$ accepts pairs of traces $\pi_1 \in (2^I)^\omega$ and $\pi_2\in (2^{I\cup O})^\omega$, where $\pi_1$ is in the cause of $E$ on $\pi_2$ in $\TT$. Afterward, $U$ can be translated to the DBW of the size $\OO(3^{|U|})$ by \cite{Miyano1984AlternatingFA}. Combining it with $\pi$ we get an automaton for the cause. We denote $\TT  = (S,s_0, AP, \delta, l)$ and $E = (Q,2^{AP}, q_0, F, \Delta)$. Let us define $U:=(S\times Q,\Sigma, s_0\times q_0, \Delta_U, S\times F)$ over the alphabet $\Sigma := 2^I \times 2^{AP}$. Here transition relation $\Delta_U$ is defined for $s\in S$, $q\in Q$, $I'_1,I'_2\subseteq I$ and $O'\subseteq O$ as follows:
    $$\Delta_U ((s,q), (I'_1, I'_2, O')) $$ $$ := \{ (s',q') \mid \exists I'\subseteq I: \ I'_1\cap I'_2 \subseteq I' \subseteq I'_1 \cup I'_2,$$ $$s'\in \delta(s,I') , \  \Delta(q,I',l(s'))=q' \} \enspace .$$
    It is easy to see that runs of $U$ on $\pi_1,\pi_2$ correspond to runs of $E$ on traces $\pi_3$ such that $\pi_3 \leq_{\pi_2} \pi_1$. Hence, $U$ accepts $\pi_1,\pi_2$ iff $E$ accepts all such $\pi_3$, which means that $\pi_2$ is in the cause of $E$ on $\pi_1$.
\end{proof}

\begin{theorem}
    For a system $\TT$, a similarity relation $\leq^{\mathit{subset}}$, a trace $\pi$ and a safety effect $E$ given as a DFW $E_{\mathit{bad\_pref}}$ for the bad prefixes of $E$, there exists a DFW for the bad prefixes of the cause of the size $\OO(|\pi| \cdot 2^{(|\TT| \cdot |E_{\mathit{bad\_pref}}| )})$.
\end{theorem}

\begin{proof}
    The proof is similar to the proof of Theorem \ref{upper_NCW}, but now we work with finite word languages.

    First, we construct the NFW $A$ over the alphabet $(2^I)^* \times (2^{I\cup O})^*$, which recognizes such pairs of words $w_1\in (2^I)^*$ and $w_2 \in (2^{I\cup O})^*$, that there exists a trace $w\leq^{subset}_{w_2} w_1$ and $w\in E_{\mathit{bad\_pref}}$. 
    
    Denote the system $\TT  = (S,s_0, AP, \delta, l)$ and the effect $E_{\mathit{bad\_pref}} = (Q,2^{AP}, q_0, F, \Delta)$ and denote the cause as $C$.$$A := (S\times Q, (2^I)^* \times (2^{I\cup O})^*, s_0\times q_0, \Delta_U, S\times F)$$ 
    The transition function $\Delta_U$ is defined as follows:
    $$\Delta_U ((s,q), (I'_1, I'_2, O')) $$ $$ := \{ (s',q') \mid \exists I'\subseteq I: \ I'_1\cap I'_2 \subseteq I' \subseteq I'_1 \cup I'_2,$$ $$s'\in \delta(s,I') , \  \Delta(q,I',l(s'))=q' \} \enspace .$$
    We build the DFW of the size $2^{|A|} = 2^{|\TT| \cdot |E_{bad\_ pref}|}$ which recognizes the same language as $A$ and combine it with $\pi$ getting the automaton $C'_{\mathit{bad\_pref}}$. Obviously, $C'_{\mathit{bad\_pref}}$ accepts only bad prefixes of $C$. But unfortunately, there may be some bad prefixes of $C$ which $C'_{\mathit{bad\_pref}}$ does not accept.  
    
    For every $\pi'' \notin C$ there exists a trace $\pi'\leq^{subset}_\pi \pi''$, such that $\pi'\notin E$. Hence, $\pi'$ has a prefix from $E_{\mathit{bad\_pref}}$. Thus, $\pi''$ has a prefix from $C'_{\mathit{bad\_pref}}$. 
    
    The last step is to denote accepting all states of $C'_{\mathit{bad\_pref}}$ from which every infinite trace eventually visits an accepting state getting the automaton $C_{\mathit{bad\_pref}}$.
\end{proof}

\begin{theorem}
    For a system $\TT$, a similarity relation $\leq^{subset}$, a trace $\pi$ and a guarantee effect $E$ given as a DFW $E_{\mathit{good\_pref}}$ for the good prefixes of $E$, there exists a DFW for the good prefixes of the cause of the size $\OO(|\pi| \cdot 2^{(|\TT| \cdot |E_{\mathit{good\_pref}}| )})$.
\end{theorem}

\begin{proof}
    By Corollary \ref{causality_preservation} the cause must also be a guarantee property. Moreover, by Theorem \ref{hierarchy_preservation} the set $C_{pairs}$ of pairs of traces $\pi_1 \in (2^I)$ and $\pi_2 \in (2^{I\cup O})^\omega$, such that $\pi_1$ is in the cause of $\pi_2$, is also a guarantee property, since it can be universally projected on the set $E$ in the similar way as in Proposition \ref{causality_is_up}. 

    First, we construct the UFW $U$ of the size $|\TT|\cdot |E_{\mathit{good\_pref}}|$ that recognizes the good prefixes of $C_{pairs}$.

    Denote $\TT  = (S,s_0, AP, \delta, l)$ and $E_{\mathit{good\_pref}} = (Q,2^{AP}, q_0, F, \Delta)$ and denote the cause as $C$.

    Let us define $U=(S\times Q,\Sigma, s_0\times q_0, \Delta_U, S\times F)$ over the alphabet $\Sigma := 2^I \times 2^{AP}$. Here transition relation $\Delta_U$ is defined for $s\in S$, $q\in Q$, $I'_1,I'_2\subseteq I$ and $O'\subseteq O$ as follows:

    $$\Delta_U ((s,q), (I'_1, I'_2, O')) $$ $$: = \{ (s',q') \mid \exists I'\subseteq I: \ I'_1\cap I'_2 \subseteq I' \subseteq I'_1 \cup I'_2,$$ $$s'\in \delta(s,I') , \  \Delta(q,I',l(s'))=q' \} \enspace .$$

    Turning this $U$ to DFW of the size $2^{|U|} = 2^{|\TT|\cdot |E_{\mathit{good\_pref}}|}$ and combining it with $\pi$ we get the DFW for the good prefixes of the cause. 
\end{proof}

\end{document}